\documentclass[a4paper,11pt,reqno]{article}
\usepackage[utf8]{inputenc} 
\usepackage{amsmath}
\usepackage{amssymb}
\usepackage{amsthm}
\usepackage{dsfont}
\usepackage[usenames,dvipsnames]{color}
\usepackage{caption}
\usepackage{subcaption}
\usepackage{tikz}
\usepackage[ruled,vlined,linesnumbered]{algorithm2e}
\usepackage[toc,page]{appendix}
\usetikzlibrary{positioning,chains,fit,shapes,calc}
\usetikzlibrary{decorations.pathreplacing}
\usetikzlibrary{topaths,calc}
\usetikzlibrary{arrows}
\usepackage{vmargin}
\usepackage{setspace}
\usepackage{lineno}
\usepackage{hyperref} 
\usepackage{url}

\definecolor{mygreen}{RGB}{40,140,120}
\definecolor{truegray}{gray}{0.75}
\definecolor{mydgreen}{RGB}{21,204,0}
\definecolor{mauve}{RGB}{51,0,102}
\definecolor{rose}{RGB}{255,0,153}
\definecolor{comms}{gray}{0.55}
\definecolor{brun}{RGB}{204,0,204}
\definecolor{moche}{RGB}{238,44,44}
\definecolor{bgray}{gray}{0.25}
\definecolor{bbgray}{gray}{0.40}
\definecolor{jaune}{RGB}{238,173,14}
\definecolor{bleu}{RGB}{58,95,205}
\definecolor{beau}{RGB}{255,203,96}
\definecolor{indigo}{RGB}{63,34,4}
\definecolor{vert2}{RGB}{58,242,75}
\definecolor{bl2}{RGB}{0,216,200}

\DeclareMathOperator{\idl}{idl}

\DeclareMathOperator{\opt}{opt}

\DeclareMathOperator{\lbl}{label}

\newcommand{\CCC}{\mathcal{C}{}}

\newcommand{\FFF}{\mathcal{F}{}}
\newcommand{\GGG}{\mathcal{G}{}}
\newcommand{\KKK}{\mathcal{K}{}}

\newcommand{\SSS}{\mathcal{S}{}}

\newcommand{\R}{\mathbb{R}}

\newtheorem{thm}{Theorem}[section]
\newtheorem{lem}{Lemma}[section]

\theoremstyle{definition}

\newtheorem{pbl}{Problem}

\theoremstyle{remark}

\begin{document}
 
\begin{spacing}{1.1}
\begin{center}
\textbf{\Huge Finding a Maximum-Weight Convex Set in a Chordal Graph}\\[1mm]
\textsc{Jean Cardinal, Jean-Paul Doignon, Keno Merckx}\\[2mm]
\end{center}

\begin{abstract}
We consider a natural combinatorial optimization problem on chordal graphs, the class of graphs with no induced cycle of length four or more.
A subset of vertices of a chordal graph is (monophonically) convex if it contains the vertices of all 
chordless paths between any two vertices of the set. The problem is to find a maximum-weight convex subset
of a given vertex-weighted chordal graph. It generalizes previously studied special cases in trees
and split graphs. It also happens to be closely related to the closure problem in partially ordered sets and directed graphs. 
We give the first polynomial-time algorithm for the problem.
\end{abstract}

\section{Introduction}

In many practical optimization problems, feasible solutions consist of one or more sets that are required to
satisfy some kind of convexity constraint. They can take the form of geometrically convex sets, such as in spatial planning
problems~\cite{Williams2003}, electoral district design~\cite{Douglas12}, or underground mine design~\cite{Parkinson2012}. 
Alternatively, convexity can be defined in a combinatorial fashion.

In the closure problem~\cite{Picard_1976}, we are given a directed graph with (positive or negative) vertex weights, and we are asked to find a maximum-weight vertex subset with no outgoing edges. In the case where the directed graph is acyclic, this amounts to find a maximum-weight downset of a partial order. Here, convexity is interpreted as the property of being downward closed. Again, many practical applications are related to the closure problem. For instance, military targeting~\cite{Orlin87}, transportation network design~\cite{ryhs70} and job scheduling~\cite{sidney75}. Recently, a parametric version of the closure problem has been studied by Eppstein~\cite{Eppstein18}.

In their seminal paper, Farber and Jamison~\cite{Farber_86} developed the foundations of a combinatorial abstraction of convexity in graphs.  In particular, they defined convex sets in graphs as subsets of vertices which contain the vertices of all chordless paths between any two vertices of the subset. This particular way of defining convexity in a graph is referred to as {\em monophonic convexity}. The collection of monophonic convex sets of a graph has specific nice properties and forms a convex geometry exactly if the graph is chordal. We consider the problem of finding a maximum-weight convex subset of a vertex-weighted chordal graph. We give a polynomial-time algorithm to solve the problem. Until now, only the special cases of trees~\cite{Wolsey95,Korte_Lovasz_1989_bis} and split graphs~\cite{merckx17} were known to be polynomial-time solvable.

Our algorithm for chordal graphs makes use of an algorithm due to Picard~\cite{Picard_1976} for the similar problem on ordered sets. Its design relies on a better understanding of the structure of a chordal graph from the point of view of its convex geometry. The results can be seen as a generalization of two algorithmic results for trees and split graphs, to all chordal graphs.

\subsection{Previous works}

The notion of a convex geometry appears in various contexts in mathematics and computer science. Dilworth~\cite{Dilworth_1940} first examined structures very close to convex geometries in terms of lattice theory. The convex geometries were formally introduced by Jamison~\cite{Jamison81,Jamison82} and Edelman and Jamison~\cite{Edelman85} in 1985. Later, Korte, Lov{\'a}sz and Schrader~\cite{Korte_Lovasz_Schrader_1991} considered antimatroids, which is the dual concept to the one of a convex geometry, as a subclass of greedoids. Today, the concept of a convex geometry (or antimatroid) appears in many fields of mathematics such as formal language theory (Boyd and Faigle~\cite{Boyd_Faigle_90}), choice theory (Koshevoy~\cite{Koshevoy_1999}), game theory (Algaba~\cite{Algaba_all_04}) and mathematical psychology (Falmagne and Doignon~\cite{Falmagne_Doignon_LS}) among others. 

When weights are assigned to the points of the convex geometry, the natural question of finding a convex set with maximum-weight arises. Particular subproblems are the closure problem~\cite{Picard_1976}, the maximum-weight subtree problem~\cite{Wolsey95,Korte_Lovasz_1989_bis}, the  maximum-weight path-closed set~\cite{Groflin84}, or in a more geometrical setting, some variants of the minimum $k$-gons problem~\cite{Eppstein92}. A more recent example is the problem of finding a maximum-weight convex set in a split graph~\cite{merckx17}. For all of these problems, polynomial-time algorithms were published. We also mention that, without focusing on algorithms, Korte and Lov{\'a}sz~\cite{Korte_Lovasz_1989_bis} gives linear characterizations of the convex set polytope for certain classes of antimatroids.

Searching for a general efficient algorithm to obtain a maximum-weight convex set in convex geometries seems hopeless because the problem is $NP$-hard even for special cases, see Eppstein~\cite{Eppstein_07} and Cardinal, Doignon and Merckx~\cite{merckx17}. However, searching for a polynomial-time algorithm for certain classes of convex geometries could reveal bridges between mathematical areas and lead to better understanding of the underlying mathematical objects.

Chordal graphs and their representations have generated an extensive literature. See for instance Blair and Peyton~\cite{Blair93}, McKee and McMorris~\cite{mckee99} or Golumbic~\cite{Golumbic2004} for theoretical and practical applications in various fields such as computational biology, phylogenetic, database, sparse matrix computation and statistics. But, despite a significant number of results about chordal graphs, there was to our knowledge, no polynomial-time algorithm to find a maximum-weight convex set.

\subsection{Structure of the paper}

In the next section, we give basic definitions and notation regarding convex geometries, graphs and posets, and formally define the optimization problem we consider. We also give the definition of the clique-separator graph of a chordal graph, which will be instrumental in what follows.
In Section~\ref{sec:step1}, we give a procedure solving the problem in a special family of instances. For this family, the problem is reduced to the closure problem in a partially ordered set. In Section~\ref{sec:step2} we generalize the algorithm to handle arbitrary chordal graphs and argue that it runs in polynomial time.

\section{Preliminaries}

We review here some basic notation and results for graphs and convex geometries, we also formally define the problems we investigate.

\subsection{Notation for graphs}

A (simple) graph $G$ is a pair $(V,E)$ where $V$ is the (finite) set of vertices and $E$ the set of edges, for a background on graph theory we recommend the book by Diestel~\cite{Diestel10}. A \emph{path} is a  sequence of distinct vertices $(v_1,\ldots,v_n)$ such that $\{v_i,v_{i+1}\}\in E$ for all $i$ in $\{1,\ldots, n-1\}$. The path is \emph{chordless} if no two vertices are connected by an edge that is not in the path. From a path $p=(v_1,\ldots,v_n)$ we can \emph{extract} a chordless path by taking a shortest path between $v_1$ and $v_n$ in the subgraph induced by the vertices in $p$. The graph is \emph{connected} if for any $u,v$ in $V$ there is a path $(u,\ldots,v)$.  A \emph{connected component} of $G$  is a maximal connected subgraph of $G$. Each vertex belongs to exactly one connected component, as does each edge. A \emph{cycle} is a path $(v_1,\ldots,v_n)$ such that $\{v_n,v_1\}$ is an edge. A cycle is \emph{chordless} if no two vertices of the cycle are connected by an edge that does not itself belong to the cycle.  A graph is \emph{chordal} if  every chordless cycle in the graph has at most three vertices. For $V'\subseteq V$ we denote by $N(V')$ the set of vertices $w$ in $V\setminus V'$ such that $\{w,v\}\in E$ for some $v$ in $V'$. We write $N(v)$ for $N(\{v\})$.

A \emph{clique} $K$ of $G$ is a set of pairwise adjacent vertices, we say that $K$ is a \emph{maximal clique} if there is no  clique $K'$ of $G$ such that $K\subset K'$. We denote by $\KKK_G$ the set of all maximal cliques in $G$. A \emph{separator} $S$ of $G$ is a set of vertices such that there exist two vertices $u,v$ in $V\setminus S$ connected by a path in the graph but not in $G-S$. We say that $S$ is a \emph{minimal separator} if there is no separator $S'$ of $G$ such that $S'\subset S$. For $u,v$ in $V$, a subset $S$ of $V\setminus \{u,v\}$ is a \emph{$uv$-separator} if $u$ and $v$ are connected in $G$ but not in $G-S$. The set $S$ is a \emph{minimal vertex separator} of $G$ if $S$ is a $uv$-separator for some $u,v$ in $V$ and $S$ does not strictly contain any $uv$-separator. Note that any minimal separator is also a minimal vertex separator, but the converse does not hold in general. We denote by $\SSS_G$ the set of all minimal vertex separators in $G$. Note that in chordal graphs, every minimal vertex separator is a clique. We observe that for any chordal graph $G=(V,E)$ we have $|\KKK_G|\leqslant |V|$ and $|\SSS_G|\leqslant |V|-1$, the proofs of those inequalities can be found in Fulkerson and Gross~\cite{fulkerson1965}, and Ho and Lee~\cite{Ho89} respectively.

\subsection{Convex geometries on posets and chordal graphs}

A set system $(V,\CCC)$, where $V$ is a finite set of elements and $ \CCC\subseteq 2^{V}$, is a \emph{convex geometry} when
\begin{align*}
&\varnothing \in \CCC,   \\
&\forall  C_1, C_2 \in \CCC : C_1\cap C_2 \in \CCC,   \\
&\forall C \in \CCC\setminus \{V\}, \, \exists \, c\in V\setminus C: C\cup \{c\} \in \CCC.  
\end{align*}
The \emph{convex sets} of the convex geometry $(V,\CCC)$ are the members of $\CCC$. The  \emph{feasible sets} are the complements in $V$ of the convex sets. An \emph{antimatroid} (or \emph{learning space}~\cite{Falmagne_Doignon_LS}) is a pair $(V,\FFF)$ such that $(V,\FFF^{\complement})$ is a convex geometry where $\FFF^{\complement}=\{V\setminus F : F\in \FFF\}$. All results on antimatroids have their counterpart for convex geometries.

We recall that a \emph{partially ordered set}  (or \emph{poset}) $P$ is a pair $(V,\leqslant)$ formed of a finite set $V$ and a binary relation $\leqslant$ over $V$ which is reflexive, antisymmetric, and transitive. For a poset  $(V,\leqslant)$ an \emph{ideal} $I$ is a subset of $V$ such that for all elements $a$ in $I$ and $b$ in $V$, if $b\leqslant a$, then $b$ is also in $I$. The ideals are also known as \emph{downsets}. We call $\idl(P)$ the set of ideals in $P$. For $u$, $v$ in $V$, we say that $v$ \emph{covers} $u$ in $P$ with $u\neq v$, if $u\leqslant v$ and there is no $x$ in $V\setminus \{u,v\}$ such that $u\leqslant x \leqslant v$.

One particular class of convex geometries described by Farber and Jamison~\cite{Farber_86}  comes from the ideals of a poset. More precisely, let $(V,\leqslant)$ be a poset, then $(V, \idl(V,\leq))$ is a  convex geometry called a \emph{downset alignment}. Thus the convex sets in $(V, \idl(V,\leq))$ are the ideals in $(V,\leqslant)$. The downset alignments  are the only convex geometries closed under union.  

For a graph $G=(V,E)$, a set $C$ of vertices is a \emph{monophonically convex set} (\emph{m-convex set}, or \emph{convex set}) if $C$ contains every vertex on every chordless path between vertices in $C$.   We denote with $\CCC_G$ the set of m-convex sets of $G$. It happens that $(V,\CCC_G)$, is a convex geometry if and only if $G$ is chordal (Farber and Jamison~\cite{Farber_86}).

Many classical problems in combinatorial optimization have the following form. For a set system $(V,\CCC)$  and for a function $w:V \rightarrow \R$, find a set $C$ of $\CCC$ maximizing the value of 
\begin{align*}
w(C)=\sum_{c\in C}w(c).
\end{align*}
For instance, the problem is known to be efficiently solvable for the system of independent sets of a matroid, thanks to the the greedy algorithm (see Oxley \cite{Oxley_2006}). Since convex geometries capture a combinatorial abstraction of convexity in the same way as matroids capture linear dependence, the question of finding a convex set of maximum-weight arises naturally.

The problem of finding efficiently a maximum-weight convex set in a poset was solved by Picard~\cite{Picard_1976}. The described algorithm calls as a subroutine a maximum flow algorithm (for instance Goldberg and Tarjan~\cite{Goldberg1988}) and runs in $O(mn \log(\frac{n^{2}}{m}))$ time, where $n$ and $m$ are respectively  the number of elements and the number of cover relations in the poset.

\subsection{The clique-separator graph for chordal graphs}
 
Ibarra~\cite{Ibarra09} introduces the \emph{clique-separator graph} for chordal graphs. For a chordal graph $G$, he defines a mixed graph where the nodes are the maximal cliques and minimal vertex separators of $G$. Moreover, the (directed) arcs and (undirected) edges respectively represent the containment relations between the maximal cliques and minimal vertex separators of $G$. The clique-separator graph $\mathcal{G}$ of a chordal graph $G$ has a set of \emph{clique nodes}, one for each clique of $G$ and a set of \emph{separator nodes} one for each minimal vertex separator of $G$. The clique-separator graph has also a set $A$ of edges and arcs defined a follow. Each arc $(S,S')$ is from a separator node $S$ to a separator node $S'$ such that $S\subset S'$ and there is no separator node $S''$ such that $S \subset S'' \subset  S'$. Each edge $\{K,S\}$ is between a clique node $K$ and a  separator node $S$ such that $S \subset K$ and there is no separator node $S'$ such that $S \subset S' \subset K$. Later in this work, we will denote by $\mathit{Ar}_{\GGG}$ the set of arcs in a clique-separator graph $\GGG$. Figure~\ref{fig:csepgr} gives us an example of a clique-separator graph of a chordal graph. Two of the mains results obtained by Ibarra~\cite{Ibarra09} are the following theorems.

\begin{thm}\label{ibarracube}
Given a chordal graph $G=(V,E)$, constructing its clique-separator graph can be done in $O(|V|^{3})$ time.
\end{thm}

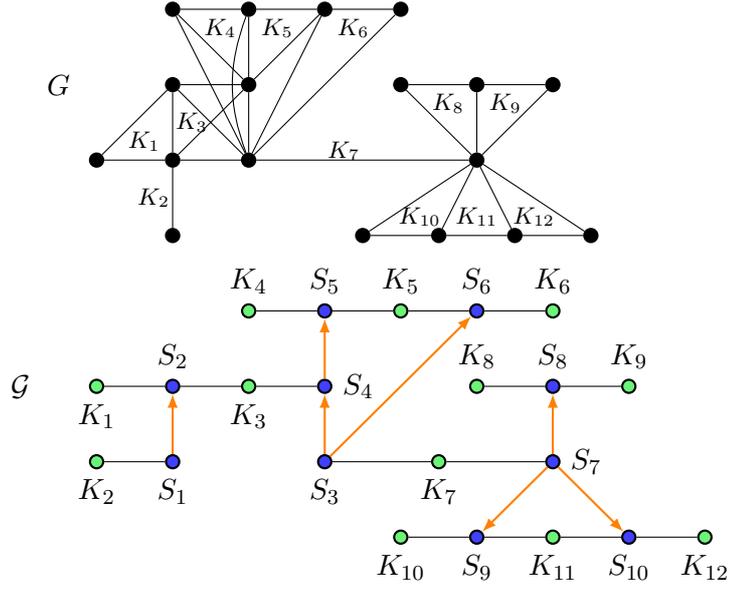
\begin{figure}[ht]
\centering
\begin{tikzpicture}[scale=0.5]
 \tikzstyle{vertex}=[circle,draw,fill=black,thick,inner sep=1.75pt]
 \draw (-1,2) node[] (l1) {$G$};
 \draw (1.25,0.5) node[] (k1) {{\footnotesize $K_1$}};
 \draw (1.5,-1) node[] (k2) {{\footnotesize $K_2$}};
 \draw (2.5,1) node[] (k3) {{\footnotesize $K_3$}}; 
 \draw (3.25,3.5) node[] (k4) {{\footnotesize $K_4$}}; 
 \draw (4.75,3.5) node[] (k5) {{\footnotesize $K_5$}}; 
 \draw (6.5,0.25) node[] (k6) {{\footnotesize $K_7$}}; 
 \draw (9.25,1.5) node[] (k7) {{\footnotesize $K_8$}};
 \draw (10.75,1.5) node[] (k8) {{\footnotesize $K_9$}};
 \draw (8.5,-1.5) node[] (k9) {{\footnotesize $K_{10}$}};
 \draw (10,-1.5) node[] (k10) {{\footnotesize $K_{11}$}};
 \draw (11.5,-1.5) node[] (k11) {{\footnotesize $K_{12}$}};
 \draw (6.75,3.5) node[] (k6) {{\footnotesize $K_6$}};
 \draw (0,0) node[vertex]  (a) {};
 \draw (2,0) node[vertex]  (b) {};
 \draw (2,-2) node[vertex]  (c) {};
 \draw (2,2) node[vertex]  (d) {};
 \draw (4,0) node[vertex]  (e) {};
 \draw (4,2) node[vertex]  (f) {};
 \draw (2,4) node[vertex]  (g) {};
 \draw (4,4) node[vertex]  (h) {};
 \draw (6,4) node[vertex]  (i) {};
 \draw (10,0) node[vertex]  (j) {};
 \draw (10,2) node[vertex]  (k) {};
 \draw (8,2) node[vertex]  (l) {};
 \draw (12,2) node[vertex]  (m) {};
 \draw (7,-2) node[vertex]  (n) {};
 \draw (9,-2) node[vertex]  (o) {};
 \draw (11,-2) node[vertex]  (p) {};
 \draw (13,-2) node[vertex]  (q) {};
 \draw (8,4) node[vertex]  (r) {};
 \draw (a)--(b)--(d)--(a) (b)--(c);
 \draw (d)--(e)--(b) (d)--(f)--(e) (f)--(b);
 \draw (f)--(h)--(g)--(f) (h)--(i)--(f) (g)--(e)--(i);
 \draw (e)--(j)--(k) (j)--(l)--(k)--(m)--(j); 
 \draw (j)--(n)--(o)--(j)--(p)--(q)--(j) (o)--(p) (i)--(r)--(e);
 \draw (h) to[in=110,out=-110] (e);
 \tikzstyle{cli}=[circle,,draw,fill=vert2!75,thick,inner sep=1.75pt]
 \tikzstyle{sep}=[circle,draw,fill=blue!75,thick,inner sep=1.75pt]
 \draw (-2,-6) node (l2a)  {$\mathcal{G}$};
 \draw (0,-8) node[cli, label=below :{$K_2$}]  (K1a) {};
 \draw (2,-8) node[sep, label=below :{$S_1$}]  (S1a) {};
 \draw (0,-6) node[cli, label=below :{$K_1$}]  (K2a) {};
 \draw (2,-6) node[sep, label=above :{$S_2$}]  (S2a) {};
 \draw (4,-6) node[cli, label=below :{$K_3$}]  (K3a) {};
 \draw (6,-6) node[sep, label=right :{$S_4$}]  (S3a) {};
 \draw (4,-4) node[cli, label=above :{$K_4$}]  (K4a) {};
 \draw (6,-4) node[sep, label=above  :{$S_5$}]  (S4a) {};
 \draw (10,-4) node[sep, label=above  :{$S_6$}]  (S6b) {};
 \draw (12,-4) node[cli, label=above  :{$K_6$}]  (K6b) {};
 \draw (8,-4) node[cli, label=above :{$K_5$}]  (K5a) {};
 \draw (6,-8) node[sep, label=below :{$S_3$}]  (S5a) {};
 \draw (9,-8) node[cli, label=below :{$K_7$}]  (K6a) {};
 \draw (12,-8) node[sep, label=right :{$S_7$}]  (S6a) {};
 \draw (12,-6) node[sep, label=above :{$S_8$}]  (S7a) {};
 \draw (10,-6) node[cli, label=above :{$K_8$}]  (K7a) {};
 \draw (14,-6) node[cli, label=above :{$K_9$}]  (K8a) {};
 \draw (12,-10) node[cli, label=below :{$K_{11}$}]  (K9a) {};
 \draw (10,-10) node[sep, label=below :{$S_{9}$}]  (S8a) {};
 \draw (14,-10) node[sep, label=below :{$S_{10}$}]  (S9a) {};
 \draw (8,-10) node[cli, label=below :{$K_{10}$}]  (K10a) {};
 \draw (16,-10) node[cli, label=below :{$K_{12}$}]  (K11a) {};
 \draw (K1a)--(S1a) (K2a)--(S2a)--(K3a)--(S3a) (K4a)--(S4a)--(K5a);
 \draw (S5a)--(K6a)--(S6a)  (K7a)--(S7a)--(K8a) (K10a)--(S8a)--(K9a)--(S9a)--(K11a);
 \draw (K5a)--(S6b)--(K6b);
 \draw[thick,-latex,orange] (S3a)->(S4a)  ;
 \draw[thick,-latex,orange] (S1a)->(S2a)  ;
 \draw[thick,-latex,orange] (S5a)->(S3a)  ;
 \draw[thick,-latex,orange] (S6a)->(S7a)  ;
 \draw[thick,-latex,orange] (S6a)->(S8a)  ;
 \draw[thick,-latex,orange] (S6a)->(S9a)  ;
 \draw[thick,-latex,orange] (S5a)->(S6b)  ;
\end{tikzpicture} 
  \caption{A clique-separator graph $\GGG$ of a chordal graph $G$}\label{fig:csepgr}
\end{figure}

\begin{thm}\label{ibarrath1p1}
Let $G=(V,E)$ be a chordal graph with clique-separator graph $\GGG$ and let $S$ be a separator node of $\GGG$. 
If $G - S$ has connected components $G_1, \ldots, G_t$, then $t>1$ and $\GGG - \{S': S'\in \SSS_G, S'\subseteq S\}$ has connected components $\GGG_1,\ldots,\GGG_t$ such that for every $1\leqslant i \leqslant t$, the vertex set of $G_i$ is the same as the vertex set represented by the nodes of $\GGG_i - S$. 
\end{thm}

\subsection{The problems}

Our main problem is to find a maximum-weight convex set in a given vertex-weighted chordal graph.  
It is the \emph{maximum-weight convex set problem in chordal graphs}.

\begin{pbl}\label{MajorPb}
Given a chordal graph $G$ and a weight function $w:V \rightarrow \R$, find a set $C$ in $\CCC_G$ that maximizes the value of $w(C)$.
\end{pbl}

Here is our main result.

\begin{thm}\label{thm:cmplxTOT}
The maximum-weight convex set problem in chordal graphs can be solved in polynomial time.
\end{thm}

The well-known problem of finding a maximum-weight connected subtree in a tree can be solved by selecting a vertex as ``root", finding a maximum-weight subtree that contains the root, and iterating this procedure for all possible roots (see Wolsey \textit{et al.}~\cite{Wolsey95}). In order to use a similar approach to solve Problem~\ref{MajorPb}, we define a notion of root. It will be easier to work with chordal graphs which are connected. Note that our results straightforwardly extend to the non-connected case.

In order to simplify some of the later statements and arguments, we want to have in each maximal clique some vertex which is adjacent to no vertex outside the clique and which has weight zero.  To this aim, we add such a vertex to any maximal clique (without changing the result of the optimization problems, see the end of the present subsection).  Formally, let $G=(V,E)$ be a vertex-weighted graph. For each maximal clique $K$ of $G$, we add a new vertex $d_K$ to the graph and we make $d_K$ adjacent to exactly the vertices in $K$.  The weight of $d_K$ is set to $0$, while the other vertices keep their weight.  The resulting vertex-weighted graph is called the \emph{extension} $G'$ of $G$.  Notice that the maximal cliques of $G'$ are all of the form $K \cup \{d_K\}$, where $K$ is a maximal clique of $G$; we call $d_K$ the \emph{dummy vertex} of the maximal clique $K \cup \{d_K\}$.
Given a vertex-weighted chordal graph $G=(V,E)$, its extension $G'=(V',E')$ is also a vertex weight chordal graph.  Remark that $G$ and $G'$ essentially have the same clique-separator graph. When $G=G'$, we say that the vertex-weighted chordal graph $G$ is \emph{extended}.

%

For a set $R$ of vertices of an extended vertex weight chordal graph~$G$, we say that a convex set $C$ of $\CCC_{G}$ is \emph{$R$-rooted} if $R\subseteq C$. 
If $R$ is a singleton $\{r\}$ we write $r$-rooted instead of $\{r\}$-rooted. This modification allows us to define the following problem.

\begin{pbl}\label{subProbs}
Given an extended chordal graph $G$ with a weight function $w:V \rightarrow \R$ 
and a maximal clique $K$ of $G$, find a $d_K$-rooted convex set $C$ of $G$ that maximizes the value of $w(C)$.
\end{pbl}

We show below that, given any vertex-weighted chordal graph $G$, solving Problem~\ref{subProbs} for the extension $G'$ of $G$ for all $K$ in $\KKK_G$ gives us a solution to Problem~\ref{MajorPb}. The first lemma states the obvious link between the convex sets of $G$ and $G'$.

\begin{lem}\label{CeitherKK}
Let $G=(V,E)$ be a chordal graph, $C$ be a convex set of $G$ and $C'$ be a convex set of $G'=(V',E')$, the extension of $G$. Then $C$ is a convex set of $G'$ and $C' \cap V$ is a convex set of $G$. 
\end{lem}

\begin{proof}
First, $C$ is a convex set of $G'$ because any chordless path in $G'$ between two vertices of $C$ is a chordless path in $G$. Second, $C' \cap V$ is convex in $G$ because any chordless path in $G$ between two vertices of $C' \cap V$ is a chordless path in $G'$.
\end{proof}

The next lemma shows a stronger result than what we need for proving the equivalence between Problem~\ref{MajorPb} and Problem~\ref{subProbs}, but it will be useful.

\begin{lem}\label{dkcupC}
Let $G=(V,E)$ be a chordal graph with a convex set $C$ in $G$, and $G'$ be the extension of $G$. Let $K_{C}$ be a maximal clique of the graph induced by $C$. Then, for every $K'$ in $\KKK_{G'}$ such that $K_{C}\subseteq K'$, the set $\{d_{K'}\} \cup C$ is convex in $G'$.
\end{lem}

\begin{proof}
For $K'$ in $\KKK_{G'}$ such that $K_{C}\subseteq K$, suppose that $\{d_{K'}\} \cup C$ is not convex in $G'$. So there is a chordless path $(d_{K'},f_1,\ldots,f_t,c)$ in $G'$ with $c$ in $C$ but $f_1,\ldots,f_t$ not in $C$.  Because $f_1$ must be in $K'$, we know that for all $v$ in $K_{C}$ we must have $\{v,c\} \in E$ (otherwise any monophonic path in $G$ we can extract from $(v,f_1,\ldots,f_t,c)$ contradicts the convexity of $C$).  There results a contradiction with the maximality of $K_{C}$.
\end{proof}

 
Lemmas~\ref{CeitherKK} and~\ref{dkcupC} combined show that any algorithm solving Problem~\ref{subProbs} in polynomial time establishes Theorem~\ref{thm:cmplxTOT}. Indeed, we run the algorithm solving Problem~\ref{subProbs} on every maximal clique and save a maximum-weight solution $C^*$ among all the outputs of the executions. Then we remove the dummy vertices from $C^*$ and we are done. 

In what follows, the chordal graphs we consider are extended:  we consider that every maximal clique $K$ contains a fixed, dummy vertex $d_K$.

\section{A special case} \label{sec:step1}

In this section, we solve Problem~\ref{subProbs} for a family of special instances. We first define a partial order relation on the vertices of a given chordal graph. Then we use this relation to reduce instances of Problem~\ref{subProbs} in this family to the closure problem in posets. The latter  problem can be solved in polynomial time using Picard's algorithm~\cite{Picard_1976}.
\subsection{The rooted poset}

Let $K$ be a maximal clique of a chordal graph $G=(V,E)$. We define the binary relation $\leqslant_K$ on $V$ as the set of pairs $(u,v)\in V \times V$ such that there is a chordless path $(v,\ldots,d_K)$  that contains $u$. For the reduction we need to check that the relation is indeed a partial order.

\begin{thm}\label{OrderOK}
For $G=(V,E)$ a chordal graph and $K$ a maximal clique of $G$, the pair $(V,\leqslant_K)$ is a poset.
\end{thm}

We give a proof for Theorem~\ref{OrderOK} in Appendix~\ref{app:proofORD}. It can be shown that the order relation we just defined is a special case of the $C$-factor relation defined by Edelman and Jamison~\cite{Edelman85} (taking the convex set $C$ equal to $\{d_K\}$). The poset $P_K=(V,\leqslant_K)$ will be referred to as the \emph{$K$-rooted poset} of $G$.  Figure~\ref{fig:Gtreat1} shows a chordal graph and the Hasse diagram for $(V,\leqslant_K)$ with $K=\{1,2,d_{\{1,2\}}\}$.

\begin{figure}[ht]
\centering
\begin{tikzpicture}[scale=0.45]
 \tikzstyle{vertex}=[circle,draw,fill=black,thick,inner sep=1.75pt]
 \tikzstyle{vertex2}=[circle,gray,draw,fill=yellow!50,thick,inner sep=1pt]
  \tikzstyle{vertex3}=[circle,draw,fill=mydgreen!75,thick,inner sep=1.75pt]
 \draw (-11,4) node[] {$G=(V,E)$};
 \draw (-2,8) node[vertex, label=right :{$1$}]  (a) {};
 \draw (-2,6) node[vertex, label=right :{$2$}]  (b) {};
 \draw (-4,6) node[vertex, label=left :{$3$}]  (c) {};
 \draw (-2,4) node[vertex, label=right :{$4$}]  (d) {};
 \draw (-4,4) node[vertex, label=left :{$5$}]  (e) {};
 \draw (-4,2) node[vertex, label=left :{$6$}]  (f) {};
 \draw (-2,2) node[vertex, label=right :{$7$}]  (g) {};
 \draw (-4,0) node[vertex, label=left :{$8$}]  (h) {};
 \draw (a)--(b)--(d)--(c)--(b);
 \draw (c)--(e)--(d) (e)--(d)--(f)--(e) (f)--(g)--(d) (f)--(h);
 \draw (-3.75,7) node[vertex2, label=left :{{\tiny $d_{\{1,2\}}$}}]  (dk1) {};
 \draw (-5.75,5) node[vertex2, label=left :{{\tiny $d_{\{3,4,5\}}$}}]  (dk2) {};
 \draw (-5.75,3) node[vertex2, label=left :{{\tiny $d_{\{4,5,7\}}$}}]  (dk3) {};
 \draw (-0.75,5) node[vertex2, label=right :{{\tiny $d_{\{2,3,4\}}$}}]  (dk4) {};
 \draw (-0.75,3) node[vertex2, label=right :{{\tiny $d_{\{4,6,7\}}$}}]  (dk5) {};
 \draw (-2.25,1) node[vertex2, label=right :{{\tiny $d_{\{7,8\}}$}}]  (dk6) {};
 \draw[gray!50] (dk4)--(c) (dk4)--(b) (dk4)--(d);
 \draw[gray!50] (dk5)--(g) (dk5)--(f) (dk5)--(d);
 \draw[gray!50] (dk2)--(c) (dk2)--(d) (dk2)--(e);
 \draw[gray!50] (dk3)--(d) (dk3)--(e) (dk3)--(f);
 \draw[gray!50] (dk1)--(a) (dk1)--(b);
 \draw[gray!50] (dk6)--(f) (dk6)--(h);
 \draw (5.5,4) node[] {$\left(V,\leqslant_{\{1,2,d_{\{1,2\}}\}}\right)$};
\draw (11,0) node[vertex3, label=below :{{\tiny $d_{\{1,2\}}$}}]  (dab2) {}; 
\draw (9,2) node[vertex3, label=below :{$1$}]  (a2) {};
\draw (13,2) node[vertex3, label=below :{$2$}]  (b2) {};
\draw (11,4) node[vertex3, label=below :{$3$}]  (c2) {};
\draw (13,4) node[vertex3, label=below right:{$4$}]  (d2) {};
\draw (15,4) node[vertex3, label=below right :{{\tiny $d_{\{2,3,4\}}$}}]  (dbcd2) {};
\draw (11,6) node[vertex3, label=below left:{{\tiny $d_{\{3,4,5\}}$}}]  (dcde2) {};
\draw (13,6) node[vertex3, label=below right :{$5$}]  (e2) {};
\draw (11,8) node[vertex3, label=below :{$6$}]  (g2) {};
\draw (15,8) node[vertex3, label=below right :{{\tiny $d_{\{4,5,7\}}$}}]  (ddeg2) {};
\draw (10,10) node[vertex3, label=below :{$7$}]  (f2) {};
\draw (8,10) node[vertex3, label=below left:{{\tiny $d_{\{7,8\}}$}}]  (dgh2) {};
\draw (12,10) node[vertex3, label=below :{$8$}]  (h2) {};
\draw (14,10) node[vertex3, label=below right  :{{\tiny $d_{\{4,6,7\}}$}}]  (ddgf2) {};
\draw (b2)--(dab2)--(a2)  (dbcd2)--(b2)--(d2) (b2)--(c2)--(dcde2)--(d2)--(e2)--(c2);
\draw (ddeg2)--(e2)--(g2)--(dgh2) (f2)--(g2)--(h2) (g2)--(ddgf2);
\end{tikzpicture}
  \caption{A chordal graph and its $\{1,2,d_{\{1,2\}}\}$-rooted poset.}\label{fig:Gtreat1}
\end{figure}
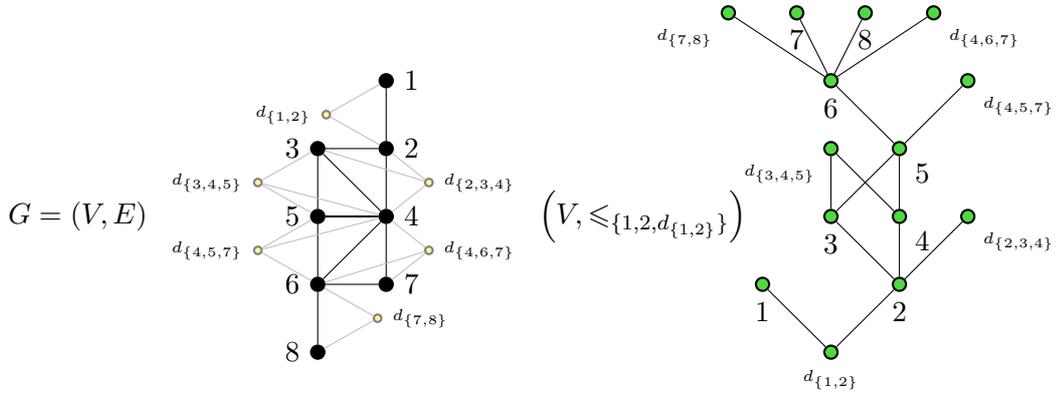

\subsection{A reduction to the maximum-weight ideal in poset problem}

We now give a sufficient condition on a pair $(G,K)$, where $G$ is a chordal graph and $K$ a maximal clique of $G$, for the existence of a 
one-to-one correspondence between the nonempty ideals of the $K$-rooted poset and the $d_K$-rooted convex sets. 
Given a chordal graph $G$ with clique-separator graph $\GGG$, for $K$ in $\KKK_G$ and  $a=(S_1,S_2)$ in  $\mathit{Ar}_{\GGG}$, we say that $a$ is \emph{$K$-blocking} if $S_1$ is a minimal $s_2d_K$-separator for every $s_2$ in $S_2\setminus S_1$. There is also an interpretation of the $K$-blocking property in the clique-separator graph. An arc $(S_1,S_2)$ is $K$-blocking if $\GGG - \{S': S'\in \SSS_G, S'\subset S_1\}$ has connected components $\GGG_1,\ldots,\GGG_t$ such that $S_2$ is included in $\GGG_i$ and $K$ in $\GGG_j$ for some distinct $i$ and $j$, and that there is no $S$ in $\GGG_j$ such that $(S,S_1)$ is an arc in $\GGG$. Figure~\ref{fig:Kblock} shows a clique-separator graph in which $(S_1,S_2)$ and $(S_6,S_7)$ are $K_1$-blocking arcs but $(S_2,S_3)$ and $(S_5,S_4)$ are not.  

\begin{figure}[ht]
\centering
\begin{tikzpicture}[scale=0.5]
 \tikzstyle{vertex}=[circle,draw,fill=black,thick,inner sep=1.75pt]
 \tikzstyle{cli}=[circle,,draw,fill=vert2!75,thick,inner sep=1.75pt]
 \tikzstyle{sep}=[circle,draw,fill=blue!75,thick,inner sep=1.75pt]
 \draw (0,0) node[cli, label=below :{$K_1$}]  (K1) {};
 \draw (2,0) node[sep, label=below :{$S_1$}]  (S1) {};
 \draw (0,2) node[cli, label=below :{$K_2$}]  (K2) {};
 \draw (2,2) node[sep, label=below left:{$S_2$}]  (S2) {};
 \draw (0,4) node[cli, label=above :{$K_3$}]  (K3) {};
 \draw (2,4) node[sep, label=above :{$S_3$}]  (S3) {};
 \draw (4,4) node[cli, label=above :{$K_4$}]  (K4) {};
 \draw (6,4) node[sep, label=above :{$S_4$}]  (S4) {};
 \draw (8,4) node[cli, label=above :{$K_5$}]  (K5) {};
 \draw (6,2) node[sep, label=below:{$S_5$}]  (S5) {};
 \draw (8,2) node[cli, label=below:{$K_6$}]  (K6) {};
 \draw (12,2)  node[sep, label=below:{$S_6$}]  (S6) {};
 \draw (14,2) node[cli, label=below:{$K_7$}]  (K7) {};
 \draw (12,4)  node[sep, label=above:{$S_7$}]  (S7) {};
 \draw (10,4) node[cli, label=above:{$K_8$}]  (K8) {};
 \draw (14,4) node[cli, label=above:{$K_9$}]  (K9) {};
 \draw (K1)--(S1) (K2)--(S2) (K3)--(S3)--(K4)--(S4)--(K5) (S5)--(K6)--(S6)--(K7) (K8)--(S7)--(K9) ;
 \draw[thick,-latex,orange] (S1)->(S2)  ;
 \draw[thick,-latex,orange] (S2)->(S3)  ;
 \draw[thick,-latex,orange] (S5)->(S4)  ;
 \draw[thick,-latex,orange] (S6)->(S7)  ;
\end{tikzpicture} 
  \caption{A clique-separator graph with exactly two $K_1$-blocking arcs.}\label{fig:Kblock}
\end{figure}
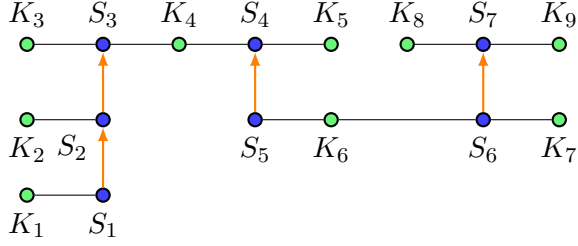

As shown in Theorem~\ref{thmCCCIDL}, the absence of $K$-blocking arcs is a sufficient condition for the correspondence between ideals of $P_K$ and $d_K$-rooted convex sets. 
Figure~\ref{fig:IdltoCCC} below gives a schematic view of the second part of the proof.

\begin{thm}\label{thmCCCIDL}
Let $G=(V,E)$ be a chordal graph with clique-separator graph $\GGG$, a maximal clique $K$ in $\KKK_{G}$ such that there is no $K$-blocking arc in $\GGG$  
and $P_K$ be the $K$-rooted poset of $G$. Then a subset $I$ of $V$ is an nonempty ideal of $P_K$ if and only if $I$ is a $d_K$-rooted convex set in $G$.
\end{thm}
\begin{proof}
First, let $C$ be a convex set containing $d_K$. For $c$ in $C\setminus \{d_K\}$, any vertex $u$ such that $u \leqslant_K c$  belongs to some chordless path. By convexity, we have $u\in C$ so $C$ is an ideal of $P_K$.

Now let $I$ be an ideal of $P_K$ and suppose, for contradiction, that $I$ is not convex. Then by definition, there must exist 
$x$ and $y$ in $I$ and a chordless path $(x,f_1,f_2,\ldots ,f_t,y)$ such that $f_1,f_2,\ldots ,f_t$ do not belong to $I$.
Note that $x$ and $y$ must be incomparable in $P_K$, as for otherwise $f_1$ or $f_t$ would be contained in $I$.
In particular, they are both different from $d_K$.
Moreover we cannot have both $\{x,d_K\}$ and $\{y,d_K\}$ as edges, since otherwise $\{x,y\}\subseteq K$. So without loss
of generality, we assume that $\{x,d_K\}\not\in E$.

Let $T$ be a minimal $xy$-separator included in the neighborhood $N(x)$ of $x$. Let $S=T\cap I$.
We claim that $S$ is either an $xd_K$-separator or a $yd_K$-separator. Suppose otherwise.
Then there must be two chordless paths of the form $(x,u_1,u_2,\ldots ,u_n,d_K)$ and $(y,v_1,v_2,\ldots ,v_{n'},d_K)$
contained in $I$ and avoiding $S$. By concatenating them, we obtain a path from $x$ to $y$ in $I$ avoiding
$S$, which contradicts the fact that $T$ was an $xy$-separator. This proves the claim.

In fact, $S$ is an $xd_K$-separator because otherwise, we can extract a chordless path from  $(y,f_t,\ldots ,f_1,x,\ldots,d_K)$ that avoids $S$ and contradicts the previous claim. 
Now consider a minimal $xd_K$-separator $S_1\subseteq S$ and a minimal vertex separator $S_2\subseteq T$ such that $a=(S_1,S_2)$ is an
arc of $\mathcal G$. We know such an arc exists because $T$ is an $xy$-separator while $S$ is not. 
We now show that $a$ is $K$-blocking, a contradiction. 

By definition, $a$ is $K$-blocking if and only if $S_1$ is a $td_K$-separator for any $t\in S_2\setminus S_1$. 
Suppose for contradiction that for some such $t$ there exists a chordless path from $t$ to $d_K$ avoiding $S_1$. 
We recall that $t\in T$ and $T\subseteq N(x)$, hence $\{x,t\}$ is in $E$.
But then there is a chordless path from $x$ to $d_K$ avoiding $S_1$, contradicting that $S_1$ is an $xd_K$-separator. 
Hence $a$ is indeed $K$-blocking.
\end{proof}
 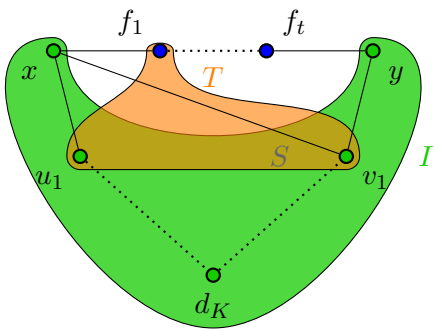
\begin{figure}[ht]
        \centering
 \begin{tikzpicture}[scale=0.7] 
   \begin{scope}[fill opacity=0.75]
    \filldraw[fill=mydgreen] ($(0,0.25)$)
    to[out=0,in=90] ($(0.25,0)$)
    to[out=-90,in=-90] ($(5.75,0)$)
    to[out=90,in=180] ($(6,0.25)$)
    to[out=0,in=0] ($(3,-5.25)$)
    to[out=180,in=180] ($(0,0.25)$);
    \end{scope}
      \begin{scope}[fill opacity=0.6]
    \filldraw[fill=orange] ($(0.25,-2)$)
    to[out=-90,in=180] ($(0.5,-2.25)$)
    to[out=-0,in=180] ($(5.5,-2.25)$)
    to[out=0,in=-90] ($(5.75,-2)$)
    to[out=90,in=-90] ($(2.25,0)$)
    to[out=90,in=90] ($(1.75,0)$)
    to[out=-90,in=90] ($(0.25,-2)$);
    \end{scope}
\draw (0,0) node[circle,draw,fill=mydgreen,thick,inner sep=1.75pt,label=below left :{$x$}] (x) {};
\draw (6,0) node[circle,draw,fill=mydgreen,thick,inner sep=1.75pt,label=below right :{$y$}] (y) {};
\draw (2,0) node[circle,draw,fill=blue,thick,inner sep=1.75pt,label=above left:{$f_1$}] (f1) {};
\draw (4,0) node[circle,draw,fill=blue,thick,inner sep=1.75pt,label=above right:{$f_t$}] (ft) {};
\draw (0.5,-2) node[circle,draw,fill=mydgreen,thick,inner sep=1.75pt,label=below left :{$u_1$}] (u1) {};
\draw (5.5,-2) node[circle,draw,fill=mydgreen,thick,inner sep=1.75pt,label=below right :{$v_1$}] (v1) {};
\draw (3,-4.25) node[circle,draw,fill=mydgreen,thick,inner sep=1.75pt,label=below :{$d_K$}] (dK) {};
\draw (u1)--(x)--(f1) (ft)--(y)--(v1)--(x) ;
 \draw[dotted, thick] (f1)--(ft);
 \draw[dotted, thick] (u1)--(dK)--(v1);
 \draw (3,-0.5) node[] {{\color{orange}  $T$}};
 \draw (7,-2) node[] {{\color{mydgreen}  $I$}};
 \draw (4.25,-2) node[] {{\color{black!60}  $S$}};
\end{tikzpicture}
        \caption{Illustration of the second part of proof for Theorem~\ref{thmCCCIDL}.}
        \label{fig:IdltoCCC}
\end{figure}

Hence whenever $G$ has no $K$-blocking arc, it is possible to compute a maximum-weight $d_K$-rooted convex set of $G$ in polynomial time by first computing the cover relation of the $K$-rooted poset, then using Picard's algorithm~\cite{Picard_1976}. Note that the relation $\leqslant_K$ can be computed in polynomial time as we show latter. There are some well-known examples of chordal graphs $G$  such that for every $K$ in $\KKK_G$, the clique-separator graph of $G$ has no $K$-blocking arc. For example, $k$-trees have no arc in their clique-separator graph (see Patil~\cite{Patil86} for details). We recall that a \emph{$k$-tree} is a graph formed by starting with a clique of size $k + 1$ and then repeatedly adding vertices with exactly $k$ neighbors inducing a clique. In the next section, we will see how to deal with the case where the clique-separator graph contains a $K$-blocking arc.

\section{A Polynomial-time Algorithm}
\label{sec:step2}

We now consider chordal graphs $G$ with one or more $K$-blocking arcs in their clique-separator graph, for some $K$ in $\KKK_G$. 
We describe an algorithm for finding a maximum-weight convex set rooted in $K$. 

For a chordal graph $G$ with clique-separator graph $\GGG$ we define the subgraph $G\ominus a$ for $a=(S_1,S_2)$ in $\mathit{Ar}_{\GGG}$ as the graph induced by the union of $S_1$ and the connected component of $G-S_1$ that intersects $S_2$. 
Figure~\ref{fig:Gominus} shows an example of the $\ominus$ operation. 
Note that $G \ominus a$  is also a chordal graph (as any induced subgraph of a chordal graph is also chordal).   

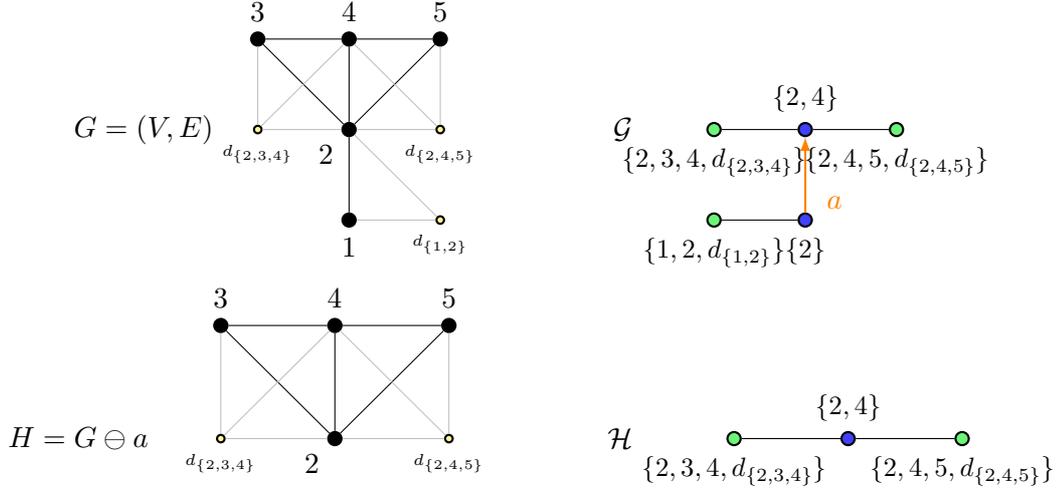
\begin{figure}[ht]
\centering
\begin{tikzpicture}[scale=0.6]
   \tikzstyle{vertex}=[circle,draw,fill=black,thick,inner sep=1.75pt]
   \tikzstyle{vertex2}=[circle,draw,fill=yellow!50,thick,inner sep=1pt]
     \tikzstyle{vertex3}=[circle,draw,fill=mydgreen!75,thick,inner sep=1.75pt]
 \draw (-4.5,2) node[] {$G=(V,E)$};
 \draw (0,0) node[vertex, label=below :{$1$}]  (r1) {};
 \draw (0,2) node[vertex, label=below left :{$2$}]  (a1) {};
 \draw (-2,4) node[vertex, label=above  :{$3$}]  (b1) {};
 \draw (0,4) node[vertex, label=above :{$4$}]  (c1) {};
 \draw (2,4) node[vertex, label=above :{$5$}]  (d1) {};
 \draw (-2,2) node[vertex2, label=below :{\tiny $d_{\{2,3,4\}}$}]  (dk1) {};
 \draw (2,2) node[vertex2, label=below :{{\tiny $d_{\{2,4,5\}}$}}]  (dk2) {};
 \draw (2,0) node[vertex2, label=below :{{\tiny $d_{\{1,2\}}$}}]  (dk3) {};
 \draw (r1)--(a1)--(b1)--(c1)--(a1)--(d1)--(c1);
 \draw[gray!50] (r1)--(dk3)--(a1);
 \draw[gray!50] (a1)--(dk2)--(c1) (dk2)--(d1);
 \draw[gray!50] (b1)--(dk1)--(a1) (dk1)--(c1);
  \tikzstyle{cli}=[circle,,draw,fill=vert2!75,thick,inner sep=1.75pt]
 \tikzstyle{sep}=[circle,draw,fill=blue!75,thick,inner sep=1.75pt]
 \draw (6,2) node (l2)  {$\mathcal{G}$};
 \draw (8,0) node[cli, label=below :{\small{$\{1,2,d_{\{1,2\}}\}$}}]  (K0) {};
 \draw (10,0) node[sep, label=below :{\small{$\{2\}$}}]  (S1) {}; 
 \draw (8,2) node[cli, label=below :{\small{$\{2,3,4,d_{\{2,3,4\}}\}$}}]  (K1) {};
 \draw (10,2) node[sep, label=above :{\small{$\{2,4\}$}}]  (S2) {};
 \draw (12,2) node[cli, label=below :{\small{$\{2,4,5,d_{\{2,4,5\}}\}$}}]  (K2) {};
 \draw (K0)--(S1) (K1)--(S2)--(K2);
 \draw[thick,-latex,orange] (S1)->(S2) node[midway,label=below right:{$a$}] {};
\end{tikzpicture} 
\begin{tikzpicture}[scale=0.75]
   \tikzstyle{vertex}=[circle,draw,fill=black,thick,inner sep=1.75pt]
   \tikzstyle{vertex2}=[circle,draw,fill=yellow!50,thick,inner sep=1pt]
     \tikzstyle{vertex3}=[circle,draw,fill=mydgreen!75,thick,inner sep=1.75pt]
 \draw (-4.5,2) node[] {$H=G\ominus a$};
 \draw (0,2) node[vertex, label=below left :{$2$}]  (a1) {};
 \draw (-2,4) node[vertex, label=above  :{$3$}]  (b1) {};
 \draw (0,4) node[vertex, label=above :{$4$}]  (c1) {};
 \draw (2,4) node[vertex, label=above :{$5$}]  (d1) {};
 \draw (-2,2) node[vertex2, label=below :{\tiny $d_{\{2,3,4\}}$}]  (dk1) {};
 \draw (2,2) node[vertex2, label=below :{{\tiny $d_{\{2,4,5\}}$}}]  (dk2) {};
 \draw (a1)--(b1)--(c1)--(a1)--(d1)--(c1);
 \draw[gray!50] (a1)--(dk2)--(c1) (dk2)--(d1);
 \draw[gray!50] (b1)--(dk1)--(a1) (dk1)--(c1);
  \tikzstyle{cli}=[circle,,draw,fill=vert2!75,thick,inner sep=1.75pt]
 \tikzstyle{sep}=[circle,draw,fill=blue!75,thick,inner sep=1.75pt]
 \draw (5,2) node (l2)  {$\mathcal{H}$};
 \draw (7,2) node[cli, label=below :{\small{$\{2,3,4,d_{\{2,3,4\}}\}$}}]  (K1) {};
 \draw (9,2) node[sep, label=above :{\small{$\{2,4\}$}}]  (S2) {};
 \draw (11,2) node[cli, label=below :{\small{$\{2,4,5,d_{\{2,4,5\}}\}$}}]  (K2) {};
 \draw   (K1)--(S2)--(K2);
\end{tikzpicture} 
  \caption{An example of the $\ominus$ operation.} \label{fig:Gominus}
\end{figure}

For a chordal graph $G=(V,E)$, a subset $R$ of $V$ and a weight function $w$, we denote by $\opt(G,R)$ a maximum-weight $R$-rooted convex set of $G$ with respect to $w$. If $R$ is a singleton $\{r\}$, we will write $\opt(G,r)$ instead of $\opt(G,\{r\})$. 
The algorithm proceeds in two main steps. In a first preprocessing phase, for each arc $a=(S_1,S_2)$, we compute $\opt(G\ominus a,S_1)$ that is, a maximum-weight convex set of $G\ominus a$ rooted in the vertex separator $S_1$. After this preprocessing phase we denote by $\lbl(a)$ the solution of this subproblem. An algorithm for this preprocessing phase is described in Section~\ref{sec:setp2PR}. 

\subsection{Computation phase}

In this second phase, we are going to use the labels of the arcs to compute a maximum-weight $d_K$-rooted convex.
The algorithm proceeds essentially by collapsing the vertices of the subgraph $(G\ominus a)-S_1$ into a single vertex $z_a$ for each arc $a=(S_1,S_2)$ that is $K$-blocking. The weight of $z_a$ is then set to $w(\lbl (a))-w(S_1)$, so that the weight of an optimal solution remains unchanged. 
This is detailed in Algorithm~\ref{Algo1}.

\begin{algorithm}[ht]
  \KwIn{a chordal graph $G$ and its clique-separator graph $\GGG$, a maximal clique $K$ of $G$, a weight function $w$, the function $\lbl$}
  \KwOut{a maximum-weight $K$-rooted convex set $C$}
 \While{$\exists \, a=(S_1,S_2) \in \mathit{Ar}_{\GGG}$ such that $a$ is  $K$-blocking}
 {
  Identify the vertices of $(G\ominus a) - S_1$ into a new vertex $z_a$\\ \label{ag1:id}
  $w(z_{a}) \gets w(\lbl(a)) - w(S_1)$\\ \label{ag1:w}
  Add a dummy vertex to the new maximal clique $\{z_{a}\}\cup S_1$\\ \label{ag1:dummy}
  Update $\GGG$\\ \label{ag1:upggg}
  }
  Use Picard's algorithm to compute a maximal weight $d_K$-rooted convex set $C$ of $G$\\
  Return $C$
 \caption{Finding a maximum $d_K$-rooted convex set in a chordal graph}
 \label{Algo1}
\end{algorithm}

Note that the number of $K$-blocking arcs decreases at each iteration of the loop. Indeed, at least the vertex separator $S_2$ disappears. 
One step of the algorithm is illustrated by Figure~\ref{fig:algo1}. 
Since the goal is to find a maximum-weight convex set in the graph, we need to remember that including the vertex $z_a$ in a solution 
for the collapsed instance amounts to choosing the set $\lbl(a)\setminus S_1$ in a solution of the original instance. 
 
 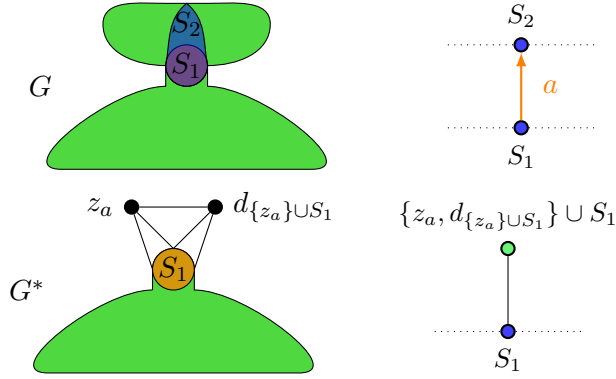
\begin{figure}[ht]
        \centering
  \begin{tikzpicture}[scale=0.55] 
 \tikzstyle{cli}=[circle,,draw,fill=vert2!75,thick,inner sep=1.75pt]
 \tikzstyle{sep}=[circle,draw,fill=blue!75,thick,inner sep=1.75pt]
   \node (a) at (0.5,0) {} ;
   \node (b) at (-2,-2) {} ;
   \node (c) at (4,-2) {} ;
   \node (d) at (1.5,0) {} ;
   \node (v1) at (-1,1.5) {} ;
   \node (v2) at (3,1.5) {} ;
   \node (p1) at (0.5,2) {} ;
   \node (p2) at (1.5,2) {} ;
   \node (pc) at (1,2) {} ;
   \node (c1) at (1,0) {} ;
   \node (c2) at (1.5,0.5) {} ;
   \node (c3) at (1,1) {} ;
   \node (c4) at (0.5,0.5) {} ;
   \begin{scope}[fill opacity=0.75]
    \filldraw[fill=mydgreen] (a.center)
    to[out=180,in=180] (b.center)
    to[out=0,in=180] (c.center)
    to[out=0,in=0] (d.center)
    to[out=90,in=-90] (c2.center)
    to[out=0,in=-90] (v2.center)
    to[out=90,in=0] (p2.center)
    to[out=180,in=0] (p1.center)
    to[out=180,in=90] (v1.center)
    to[out=-90,in=180] (c4.center)
    to[out=-90,in=90] (a.center);
    \end{scope}
      \begin{scope}[fill opacity=0.76]
    \filldraw[fill=orange] (c1.center)
    to[out=0,in=-90] (c2.center)
    to[out=90,in=0] (c3.center)
    to[out=180,in=90] (c4.center)
    to[out=-90,in=180] (c1.center);
    \end{scope}
         \begin{scope}[fill opacity=0.5]
    \filldraw[fill=blue] (c4.center)
    to[out=-90,in=180] (c1.center)
    to[out=0,in=-90] (c2.center)
    to[out=100,in=-40] (pc.center)
    to[out=220,in=80] (c4.center);
    \end{scope}
    \draw (1,0.5) node[] {{\color{black}  $S_1$}};
    \draw (1,1.5) node[] {{\color{black}  $S_2$}};
    \draw (-2.5,0) node[] {{\color{black}  $G$}};
    \node (g1) at (7,-1) {} ;
    \draw (9,-1) node[sep, label=below :{\small{$S_1$}}]  (S1) {}; 
    \node (g2) at (11,-1) {} ;
    \node (g3) at (7,1) {} ;
    \draw (9,1) node[sep, label=above :{\small{$S_2$}}]  (S2) {}; 
    \node (g4) at (11,1) {} ;
    \draw[dotted] (g1)--(S1)--(g2) (g3)--(S2)--(g4);
     \draw[thick,-latex,orange] (S1)->(S2)node[midway,label=right:{$a$}] {};
\end{tikzpicture}
\hspace{1cm}
  \begin{tikzpicture}[scale=0.55] 
 \tikzstyle{cli}=[circle,,draw,fill=vert2!75,thick,inner sep=1.75pt]
 \tikzstyle{sep}=[circle,draw,fill=blue!75,thick,inner sep=1.75pt]
  \tikzstyle{vertex}=[circle,draw,fill=black,thick,inner sep=1.75pt]
   \node (a) at (0.5,0) {} ;
   \node (b) at (-2,-2) {} ;
   \node (c) at (4,-2) {} ;
   \node (d) at (1.5,0) {} ;
   \node (v1) at (-1,1.5) {} ;
   \node (v2) at (3,1.5) {} ;
   \node (p1) at (0.5,2) {} ;
   \node (p2) at (1.5,2) {} ;
   \node (c1) at (1,0) {} ;
   \node (c2) at (1.5,0.5) {} ;
   \node (c3) at (1,1) {} ;
   \node (c4) at (0.5,0.5) {} ;
   \begin{scope}[fill opacity=0.75]
    \filldraw[fill=mydgreen] (a.center)
    to[out=180,in=180] (b.center)
    to[out=0,in=180] (c.center)
    to[out=0,in=0] (d.center)
    to[out=90,in=-90] (c2.center)
    to[out=90,in=0] (c3.center)
    to[out=180,in=90] (c4.center)
    to[out=-90,in=90] (a.center);
    \end{scope}
      \begin{scope}[fill opacity=0.76]
    \filldraw[fill=orange] (c1.center)
    to[out=0,in=-90] (c2.center)
    to[out=90,in=0] (c3.center)
    to[out=180,in=90] (c4.center)
    to[out=-90,in=180] (c1.center);
    \end{scope}
    \draw (1,0.5) node[] {{\color{black}  $S_1$}};
    \draw (-2.5,0) node[] {{\color{black}  $G^\ast$}};
    \node (g1) at (7,-1) {} ;
    \draw (9,-1) node[sep, label=below :{\small{$S_1$}}]  (S1) {}; 
    \node (g2) at (11,-1) {} ;
    \node (g3) at (7,1) {} ;
    \draw (9,1) node[cli, label=above :{\small{$\{z_{a},d_{\{z_a\}\cup S_1}\}\cup S_1$}}]  (S2) {}; 
    \node (g4) at (11,1) {} ;
    \draw[dotted] (g1)--(S1)--(g2) ;
    \draw  (S1)--(S2) {};
    \draw (2,2) node[vertex, label=right :{$d_{\{z_a\}\cup S_1}$}]  (1) {};
    \draw (0,2) node[vertex, label=left :{$z_{a}$}]  (2) {};
    \draw (1)--(2) (c3.center)--(1)--(c2.center) (c4.center)--(2)--(c3.center);
\end{tikzpicture}
        \caption{Illustration of the transformation in Algorithm~\ref{Algo1} and the implication for the clique-separator graph.}
        \label{fig:algo1}
\end{figure}

\begin{thm}\label{thm_agr}
Let $G=(V,E)$ be a chordal graph with a maximal clique $K$ and let $a=(S_1,S_2)$ be a $K$-blocking arc of $\mathit{Ar}_{\GGG}$. Let $G^\ast$ be the graph obtained from $G$ after applying Steps \ref{ag1:id}--\ref{ag1:dummy} of Algorithm~\ref{Algo1} on $a$. Then $w(\opt(G,d_K)) = w(\opt(G^\ast,d_K))$.
\end{thm}
 
Before proving Theorem~\ref{thm_agr}, we make two simple observations.

\begin{lem}\label{C1cupC2}
Let $G$ be a chordal graph,  $S$ be a minimal vertex separator of $G$ and let $V_1$ and $V_2$ be the vertex sets of two distinct components of $G - S$. If $C_1$ and $C_2$ are two $S$-rooted convex sets in the graphs induced by $V_1\cup S$ and $V_2\cup S$  respectively, then $C_1\cup C_2$ is a convex set of $G$.
\end{lem}
\begin{proof}
By contradiction, suppose there are $c$ and $c'$ in $C_1\cup C_2$ and a chordless path $(c,f_1,\ldots,f_n,c')$ of $G$ with $f_1,\ldots,f_n$ outside of $C_1\cup C_2$. There must exist $i$ in $\{1,\ldots,n\}$ such that $f_i$ is in $S$ otherwise we have a contradiction with the convexity of $C_1$ or $C_2$. But then, $f_i \in C_1\cup C_2$ because $S\subseteq C_1\cap C_2$, and we have a contradiction.
\end{proof}

\begin{lem}\label{G1G2}
Let $G=(V,E)$ be a chordal graph with a maximal clique $K$ and let $a=(S_1,S_2)$ be a $K$-blocking arc of $\mathit{Ar}_{\GGG}$. Then, for a $d_K$-rooted convex set $C$ in $G$  that contains some vertex of $(G\ominus a)- S_1$, we have $S_1\subset C$.
\end{lem}
\begin{proof}
By contradiction, let $s_1$ be in $S_1\setminus C$ and let $c$ be in $C\cap (G\ominus a - S_1)$. We know that $d_K$ and $c$ are not in the same connected component of $G - S_1$. Because $a$ is $K$-blocking, there is a chordless path $(d_K,v_1,\ldots,v_n,s_1,s_2)$ with $s_2\in S_2\setminus S_1$, $s_1\in S_1$ and $v_n\notin S_1$. There is also a path $(s_1,s'_2,\ldots,c)$ in $G\ominus a$ with $s'_2\in S_2\setminus S_1$ from which we can extract a chordless path that only intersects $S_1$ in $s_1$. So we build a path $(d_K,v_1,\ldots,v_n,s_1,\ldots,c)$ that can not have a chord, a contradiction.
\end{proof}

\begin{proof}[Proof of Theorem \ref{thm_agr}]
We decompose the equality into two inequalities. First, we show $w(\opt(G,d_K)) \geqslant w(\opt(G^\ast,d_K))$. More precisely, we show that for every $d_K$-rooted convex $C^\ast$ of $G^\ast$, we have a $d_K$-rooted convex set $C$ of  $G$ with $w(C) = w(C^\ast)$. If $z_{a}\notin C$ we take $C^\ast=C$ and we are done. Now, if $z_{a}\in C^\ast$, we define $C$ as the union of $C^\ast\setminus \{z_{a}\}$ (which is convex because $z_{a}$ is simplicial) with $\lbl(a)$. Obviously $w(C)=w(C^\ast)$, we only need to check that $C^\ast$ is convex. Because $a$ is $K$-blocking, $N(z_{a})\setminus \{d_{\{z_a\}\cup S_1}\}=S_1$ is a minimal $d_Kz_{a}$-separator, hence $C\setminus \{z_a\}$ must includes $S_1$. We also know that $S_1$ is contained in $\lbl(a)$. From Lemma~\ref{C1cupC2}, $C$ is convex in $G$, therefore $w(\opt(G,d_K)) \geqslant w(\opt(G^\ast,d_K))$.

We now show $w(\opt(G,d_K)) \leqslant w(\opt(G^\ast,d_K))$. More precisely, for every $d_K$-rooted convex $C$ of $G$, we have a $d_K$-rooted convex set $C^\ast$ of  $G^\ast$ with $w(C)\leqslant w(C^\ast)$. If $C$ does not intersect $(G\ominus a) - S_1$, we take $C=C^\ast$ and we are done. If $C$ does intersect $(G\ominus a) - S_1$,  we define $C^\ast$ as the union of ${z_{a}}$ with the vertices of $C$ that also are in $G^\ast$. We have $w(C)\leqslant w(C^\ast)$, otherwise we contradict the maximality of $\lbl(a)$. From Lemma~\ref{G1G2}, the vertices of $C$ that are also in $G^\ast$ form a convex set containing $S_1$. Since $\{z_{a}\}\cup S_1$ is a clique, hence is convex, Lemma~\ref{C1cupC2} implies that $C^\ast$ is also convex. 
So we have $w(\opt(G,d_K)) \leqslant w(\opt(G^\ast,d_K))$.
\end{proof}

\subsection{Preprocessing}
\label{sec:setp2PR}
We now describe the algorithm for computing the labels for the arcs in $\GGG$.  
This step is done only once and does not depend on the root of the convex set we are looking for. 
Recall that the label of an arc $a=(S_1,S_2)$ is the maximum-weight convex set of $G\ominus a$ rooted in $S_1$. 
Note that this algorithm uses Algorithm~\ref{Algo1} as a subroutine on smaller graphs.

The algorithm is composed of two main ingredients. First, we need to label the arcs in an order such that the computation only involves arcs that are already labeled. We prove that we can achieve this by following the order of inclusion of the graphs $G\ominus a$. Second, in order to compute the optimal convex set rooted in $S_1$, we need to check all possible roots $d_K$ such that $S_1$ is contained in $K$.
This is detailed in Algorithm~\ref{Algo2}.

\begin{algorithm}[ht]
  \KwIn{a chordal graph $G$ and its clique separator graph $\GGG$, a maximal clique $K$ of $G$, a weight function $w$}
  \KwOut{the $\lbl$ function}
 \While{$\exists$ an arc in $\mathit{Ar}_{\GGG}$ without label}
 {
  Select $a=(S_1,S_2)  \in \mathit{Ar}_{\GGG}$ without label such that every arc $a'$ with $G\ominus a' \subset G\ominus a$ is already labeled\\  \label{ag2:selec}
   $M\gets S_1$\\ \label{ag2:M1}
   Let $\KKK_a$ be the set of maximal cliques of $G$ that contain $S_1$ and are contained in $G\ominus a$\\ \label{ag2:Ka}
 \For{$K$ in $\KKK_a$}{
 Using Algorithm~\ref{Algo1}, compute a maximum-weight convex set $C^{*}$ of $G\ominus a$ rooted in $d_K$ and containing $S_1$\\  \label{ag2:al1}
 $M\gets \max_{w}\{M,C^{*}\}$\\ \label{ag2:max}
}
$\mathit{label}(a) \gets M$
}  
 \caption{Labeling the arcs in $\GGG$}
 \label{Algo2}
\end{algorithm}

In step~\ref{ag2:al1} of Algorithm~\ref{Algo2}, we can force $S_1$ to be in the solution $C^{*}$ by assigning sufficiently large weight to each vertex of $S_1$ before calling Algorithm~\ref{Algo1}. More precisely, we assign them the weight $\sum_{v\in V} |w(v)|$. By Lemma~\ref{dkcupC}, looking for all the $d_K$-rooted convex sets with $K$ in $\KKK_a$ ensures that we will find a maximum-weight $S_1$-rooted convex set of $G\ominus a$. 

Note that Algorithm~\ref{Algo2} labels the arcs in an order compatible with the partial order of inclusion of the graphs $G\ominus a$. 
The following lemma guarantees that the $K$-blocking arcs that will be processed by Algorithm~\ref{Algo1} are all already labeled.


\begin{lem}\label{TreatSort}
Let $G$ be a chordal graph with clique-separator graph $\GGG$ and let $a=(S_1,S_2)$ in $\mathit{Ar}_{\GGG}$ and $a'=(S_3,S_4)$ an arc of the clique-separator graph of $G\ominus a$. If $G\ominus a'\not\subset G\ominus a$, then $a'$ is not  $K$-blocking for any maximal clique $K$ in $\KKK_a$.
\end{lem}
\begin{proof}
Suppose that $G\ominus a'\not\subset G\ominus a$. We show that $S_3$ is not a minimal $s_4d_K$-vertex separator for $s_4$ in $S_4\setminus S_3$. If $G\ominus a' = G\ominus a$, then $d_K$ is connected to $s_4$ in $(G\ominus a') - S_3$ and we have the result.  

If $G\ominus a' \neq G\ominus a$,  there is $v$ in $G\ominus a'$ such that $v$ is not in $G\ominus a$. So there must exist a chordless path $p$ from $s_4$ to $v$ that avoids $S_3$. But, because $S_4$ is in  $G\ominus a$ and $v$ is not, the path $p$ must contains a vertex $s_1\in S_1$.  Now, because $\{s_1,d_K\}$ is an edge in $G$, we can deduce the existence of a path from $s_4$ to $d_K$ that avoids $S_3$.
\end{proof} 
A complete execution of the algorithm on an example is given in Appendix~\ref{app:example}.
%

\subsection{Time complexity}

From Algorithms~\ref{Algo1} and~\ref{Algo2}, it seems straightforward that the time complexity needed to solve the maximum-weight convex set problem on a chordal graph $G=(V,E)$ is bounded by a polynomial in $|V|$ and $|E|$. More precisely, we see that the complexity of the algorithm used to solve Problem~\ref{MajorPb} on $G$ will be bounded by that of the preprocessing step. Indeed, the preprocessing step involves $|V||E|$ calls to Picard's algorithm. We recall that the time complexity of Picard's algorithm in our case is $O(|V||E|\log(\frac{|V|^{2}}{|E|}))$.
Hence the overall running time of our algorithm is $O(|V|^{2}|E|^{2}\log(\frac{|V|^{2}}{|E|}))$. If we denote by $n$ the number of vertices of the input graph, then this running time is $O(n^6\log n)$.

To prove that the maximum-weight convex set problem on a chordal graph can be solved in this running time we need to show that all the information we need in Algorithms~\ref{Algo1} and~\ref{Algo2} can be computed in a time bounded asymptotically by the time of the preprocessing step. More precisely, given the chordal graph $G$, we can compute the following information in $O(|V|^{2}|E|^{2}\log(\frac{|V|^{2}}{|E|}))$ time: 
the clique-separator graph $\GGG$ of $G$, 
the vertices in $G\ominus a$ for each $a$ in $\mathit{Ar}_{\GGG}$,
the cliques in $\KKK_G$ for which $a$ is $K$-blocking for each $a$  in $\mathit{Ar}_{\GGG}$,
the clique in $\KKK_a$ for each $a$  in $\mathit{Ar}_{\GGG}$, 
the matrix of the relations $\leqslant_K$ for each $K$ in $\KKK_G$ and
a total order on the arcs $\GGG$ such that $G\ominus a \subseteq G\ominus a'$ implies $a<a'$ for $a$, $a'$ in $\mathit{Ar}_{\GGG}$.

The detailed proofs are given in Appendix~\ref{app:cmplx}.
This concludes the proof of Theorem~\ref{thm:cmplxTOT}. 

\newpage

\bibliographystyle{plainurl}
\bibliography{bibGLOBAL.bib}

\begin{thebibliography}{10}

\bibitem{Wolsey95}
El~Houssaine Aghezzaf, Thomas~L. Magnanti, and Laurence~A. Wolsey.
\newblock Optimizing constrained subtrees of trees.
\newblock {\em Math. Programming}, 71(2, Ser. A):113--126, 1995.
\newblock URL: \url{https://doi.org/10.1007/BF01585993}, \href
  {http://dx.doi.org/10.1007/BF01585993} {\path{doi:10.1007/BF01585993}}.

\bibitem{Algaba_all_04}
Encarnación~Algaba Algaba, Jesús~Mario Bilbao, René van~den Brink, and
  Andrés Jiménez-Losada.
\newblock Cooperative games on antimatroids.
\newblock {\em Discrete Mathematics}, 282(1-3):1--15, 2004.
\newblock URL:
  \url{http://dblp.uni-trier.de/db/journals/dm/dm282.html#DuranBBJ04}.

\bibitem{Blair93}
Jean R.~S. Blair and Barry Peyton.
\newblock An introduction to chordal graphs and clique trees.
\newblock In {\em Graph theory and sparse matrix computation}, volume~56 of
  {\em IMA Vol. Math. Appl.}, pages 1--29. Springer, New York, 1993.
\newblock URL: \url{http://dx.doi.org/10.1007/978-1-4613-8369-7_1}, \href
  {http://dx.doi.org/10.1007/978-1-4613-8369-7_1}
  {\path{doi:10.1007/978-1-4613-8369-7_1}}.

\bibitem{Boyd_Faigle_90}
E.~Andrew Boyd and Ulrich Faigle.
\newblock An algorithmic characterization of antimatroids.
\newblock {\em Discrete Appl. Math.}, 28(3):197--205, 1990.
\newblock URL: \url{http://dx.doi.org/10.1016/0166-218X(90)90002-T}, \href
  {http://dx.doi.org/10.1016/0166-218X(90)90002-T}
  {\path{doi:10.1016/0166-218X(90)90002-T}}.

\bibitem{merckx17}
Jean Cardinal, Jean-Paul Doignon, and Keno Merckx.
\newblock {On the shelling antimatroids of split graphs}.
\newblock {\em {Discrete Mathematics \& Theoretical Computer Science}}, {Vol 19
  no. 1}, March 2017.
\newblock URL: \url{http://dmtcs.episciences.org/3201}.

\bibitem{Diestel10}
Reinhard Diestel.
\newblock {\em Graph theory}, volume 173 of {\em Graduate Texts in
  Mathematics}.
\newblock Springer, Heidelberg, fourth edition, 2010.
\newblock URL: \url{http://dx.doi.org/10.1007/978-3-642-14279-6}, \href
  {http://dx.doi.org/10.1007/978-3-642-14279-6}
  {\path{doi:10.1007/978-3-642-14279-6}}.

\bibitem{Dilworth_1940}
Robert~P. Dilworth.
\newblock Lattices with unique irreducible decompositions.
\newblock {\em Ann. of Math. (2)}, 41:771--777, 1940.

\bibitem{Edelman85}
Paul~H. Edelman and Robert~E. Jamison.
\newblock The theory of convex geometries.
\newblock {\em Geom. Dedicata}, 19(3):247--270, 1985.
\newblock URL: \url{http://dx.doi.org/10.1007/BF00149365}, \href
  {http://dx.doi.org/10.1007/BF00149365} {\path{doi:10.1007/BF00149365}}.

\bibitem{Eppstein_07}
David Eppstein.
\newblock Pruning antimatroids is hard, 2007.
\newblock URL: \url{http://11011110.livejournal.com/91976.html}.

\bibitem{Eppstein18}
David Eppstein.
\newblock The parametric closure problem.
\newblock {\em {ACM} Trans. Algorithms}, 14(1):2:1--2:22, 2018.
\newblock URL: \url{http://doi.acm.org/10.1145/3147212}, \href
  {http://dx.doi.org/10.1145/3147212} {\path{doi:10.1145/3147212}}.

\bibitem{Eppstein92}
David Eppstein, Mark Overmars, G\"unter Rote, and Gerhard Woeginger.
\newblock Finding minimum area {$k$}-gons.
\newblock {\em Discrete Comput. Geom.}, 7(1):45--58, 1992.
\newblock URL: \url{https://doi.org/10.1007/BF02187823}, \href
  {http://dx.doi.org/10.1007/BF02187823} {\path{doi:10.1007/BF02187823}}.

\bibitem{Falmagne_Doignon_LS}
Jean-Claude. Falmagne and Jean-Paul Doignon.
\newblock {\em Learning Spaces}.
\newblock Springer-Verlag, Berlin, 2011.

\bibitem{Farber_86}
Martin Farber and Robert~E. Jamison.
\newblock Convexity in graphs and hypergraphs.
\newblock {\em SIAM J. Algebraic Discrete Methods}, 7(3):433--444, 1986.
\newblock URL: \url{http://dx.doi.org/10.1137/0607049}, \href
  {http://dx.doi.org/10.1137/0607049} {\path{doi:10.1137/0607049}}.

\bibitem{fulkerson1965}
Delbert~R. Fulkerson and Oliver~A. Gross.
\newblock Incidence matrices and interval graphs.
\newblock {\em Pacific J. Math.}, 15(3):835--855, 1965.
\newblock URL: \url{http://projecteuclid.org/euclid.pjm/1102995572}.

\bibitem{Goldberg1988}
Andrew~V. Goldberg and Robert~E. Tarjan.
\newblock A new approach to the maximum-flow problem.
\newblock {\em J. Assoc. Comput. Mach.}, 35(4):921--940, 1988.
\newblock URL: \url{http://dx.doi.org/10.1145/48014.61051}, \href
  {http://dx.doi.org/10.1145/48014.61051} {\path{doi:10.1145/48014.61051}}.

\bibitem{Golumbic2004}
Martin~Charles Golumbic.
\newblock {\em Algorithmic graph theory and perfect graphs}, volume~57 of {\em
  Annals of Discrete Mathematics}.
\newblock Elsevier Science B.V., Amsterdam, second edition, 2004.
\newblock With a foreword by Claude Berge.

\bibitem{Groflin84}
Heinz Gr\"oflin.
\newblock Path-closed sets.
\newblock {\em Combinatorica}, 4(4):281--290, 1984.
\newblock URL: \url{https://doi.org/10.1007/BF02579138}, \href
  {http://dx.doi.org/10.1007/BF02579138} {\path{doi:10.1007/BF02579138}}.

\bibitem{Ho89}
Chin~Wen Ho and Richard C.~T. Lee.
\newblock Counting clique trees and computing perfect elimination schemes in
  parallel.
\newblock {\em Inform. Process. Lett.}, 31(2):61--68, 1989.
\newblock URL: \url{http://dx.doi.org/10.1016/0020-0190(89)90070-7}, \href
  {http://dx.doi.org/10.1016/0020-0190(89)90070-7}
  {\path{doi:10.1016/0020-0190(89)90070-7}}.

\bibitem{Hopcroft73}
John Hopcroft and Robert Tarjan.
\newblock Algorithm 447: Efficient algorithms for graph manipulation.
\newblock {\em Commun. ACM}, 16(6):372--378, June 1973.
\newblock URL: \url{http://doi.acm.org/10.1145/362248.362272}, \href
  {http://dx.doi.org/10.1145/362248.362272} {\path{doi:10.1145/362248.362272}}.

\bibitem{Ibarra09}
Louis Ibarra.
\newblock The clique-separator graph for chordal graphs.
\newblock {\em Discrete Appl. Math.}, 157(8):1737--1749, 2009.
\newblock URL: \url{http://dx.doi.org/10.1016/j.dam.2009.02.006}, \href
  {http://dx.doi.org/10.1016/j.dam.2009.02.006}
  {\path{doi:10.1016/j.dam.2009.02.006}}.

\bibitem{Jamison81}
Robert~E. Jamison-Waldner.
\newblock Partition numbers for trees and ordered sets.
\newblock {\em Pacific J. Math.}, 96(1):115--140, 1981.
\newblock URL: \url{http://projecteuclid.org/euclid.pjm/1102734951}.

\bibitem{Jamison82}
Robert~E. Jamison-Waldner.
\newblock A perspective on abstract convexity: classifying alignments by
  varieties.
\newblock In {\em Convexity and related combinatorial geometry ({N}orman,
  {O}kla., 1980)}, volume~76 of {\em Lecture Notes in Pure and Appl. Math.},
  pages 113--150. Dekker, New York, 1982.

\bibitem{Douglas12}
Douglas~M. King, Sheldon~H. Jacobson, Edward~C. Sewell, and Wendy K.~Tam Cho.
\newblock Geo-graphs: An efficient model for enforcing contiguity and hole
  constraints in planar graph partitioning.
\newblock {\em Operations Research}, 60(5):1213--1228, 2012.
\newblock URL: \url{https://doi.org/10.1287/opre.1120.1083}, \href
  {http://dx.doi.org/10.1287/opre.1120.1083}
  {\path{doi:10.1287/opre.1120.1083}}.

\bibitem{Korte_Lovasz_1989_bis}
Bernhard Korte and L{\'a}szl{\'o} Lov{\'a}sz.
\newblock Polyhedral results for antimatroids.
\newblock {\em Annals New York Academy of Sciences}, 555:283--295, 1989.

\bibitem{Korte_Lovasz_Schrader_1991}
Bernhard Korte, L{\'a}szl{\'o} Lov{\'a}sz, and Rainer Schrader.
\newblock {\em Greedoids}, volume~4 of {\em Algorithms and Combinatorics}.
\newblock Springer-Verlag, Berlin, 1991.
\newblock URL: \url{http://dx.doi.org/10.1007/978-3-642-58191-5}, \href
  {http://dx.doi.org/10.1007/978-3-642-58191-5}
  {\path{doi:10.1007/978-3-642-58191-5}}.

\bibitem{Koshevoy_1999}
Gleb~A. Koshevoy.
\newblock Choice functions and abstract convex geometries.
\newblock {\em Math. Social Sci.}, 38(1):35--44, 1999.
\newblock URL: \url{http://dx.doi.org/10.1016/S0165-4896(98)00044-4}, \href
  {http://dx.doi.org/10.1016/S0165-4896(98)00044-4}
  {\path{doi:10.1016/S0165-4896(98)00044-4}}.

\bibitem{mckee99}
Terry~A. McKee and F.~R. McMorris.
\newblock {\em Topics in intersection graph theory}.
\newblock SIAM Monographs on Discrete Mathematics and Applications. Society for
  Industrial and Applied Mathematics (SIAM), Philadelphia, PA, 1999.
\newblock URL: \url{https://doi.org/10.1137/1.9780898719802}.

\bibitem{Orlin87}
D.~Orlin.
\newblock Optimal weapons allocation against layered defenses.
\newblock {\em Naval Research Logistics (NRL)}, 34(5):605--617, 1987.
\newblock URL:
  \url{https://onlinelibrary.wiley.com/doi/abs/10.1002/1520-6750\%28198710\%2934\%3A5\%3C605\%3A\%3AAID-NAV3220340502\%3E3.0.CO\%3B2-L},
  \href
  {http://dx.doi.org/10.1002/1520-6750(198710)34:5<605::AID-NAV3220340502>3.0.CO;2-L}
  {\path{doi:10.1002/1520-6750(198710)34:5<605::AID-NAV3220340502>3.0.CO;2-L}}.

\bibitem{Oxley_2006}
James Oxley.
\newblock {\em Matroid theory}, volume~21 of {\em Oxford Graduate Texts in
  Mathematics}.
\newblock Oxford University Press, Oxford, second edition, 2011.
\newblock URL:
  \url{http://dx.doi.org/10.1093/acprof:oso/9780198566946.001.0001}, \href
  {http://dx.doi.org/10.1093/acprof:oso/9780198566946.001.0001}
  {\path{doi:10.1093/acprof:oso/9780198566946.001.0001}}.

\bibitem{Parkinson2012}
Anita~Frances Parkinson.
\newblock {\em Essays on sequence optimization in block cave mining and
  inventory policies with two delivery sizes}.
\newblock PhD thesis, University of British Columbia, 2012.
\newblock URL:
  \url{https://open.library.ubc.ca/cIRcle/collections/24/items/1.0073023},
  \href {http://dx.doi.org/http://dx.doi.org/10.14288/1.0073023}
  {\path{doi:http://dx.doi.org/10.14288/1.0073023}}.

\bibitem{Patil86}
H.~P. Patil.
\newblock On the structure of {$k$}-trees.
\newblock {\em J. Combin. Inform. System Sci.}, 11(2-4):57--64, 1986.

\bibitem{Picard_1976}
Jean-Claude Picard.
\newblock Maximal closure of a graph and applications to combinatorial
  problems.
\newblock {\em Management Sci.}, 22(11):1268--1272, 1975/76.

\bibitem{ryhs70}
J.~M.~W. Rhys.
\newblock A selection problem of shared fixed costs and network flows.
\newblock {\em Management Science}, 17(3):200--207, 1970.
\newblock URL:
  \url{https://EconPapers.repec.org/RePEc:inm:ormnsc:v:17:y:1970:i:3:p:200-207}.

\bibitem{sidney75}
Jeffrey~B. Sidney.
\newblock Decomposition algorithms for single-machine sequencing with
  precedence relations and deferral costs.
\newblock {\em Operations Research}, 23(2):283--298, 1975.
\newblock URL: \url{http://www.jstor.org/stable/169530}.

\bibitem{Williams2003}
Justin~C Williams.
\newblock Convex land acquisition with zero-one programming.
\newblock {\em Environment and Planning B: Planning and Design},
  30(2):255--270, 2003.
\newblock URL: \url{https://doi.org/10.1068/b12925}, \href
  {http://arxiv.org/abs/https://doi.org/10.1068/b12925}
  {\path{arXiv:https://doi.org/10.1068/b12925}}, \href
  {http://dx.doi.org/10.1068/b12925} {\path{doi:10.1068/b12925}}.

\end{thebibliography}

\appendix

\section{Proof of Theorem~\ref{OrderOK}}
\label{app:proofORD}

Theorem~\ref{OrderOK} directly follows from the two lemmas below, respectively stating that the relation is antisymmetric and transitive, and whose proofs are illustrated in Figure~\ref{fig:antitran}. The reflexivity of the relation is obvious. 

\begin{lem}[Antisymmetry]
\label{AntisymOK}
For $G=(V,E)$ a chordal graph and $K$ a maximal clique of $G$, the relation $\leqslant_K$ is antisymmetric.
\end{lem}
\begin{proof}
For $a$ and $b$ in $V$, we show that we cannot have $a\leqslant_K b$, $b\leqslant_K a$ and $a\neq b$. Suppose $a\leqslant_K b$ and $b\leqslant_K a$, so there are two chordless paths $(u_1,\ldots, u_j, \ldots, u_n)$ and  $(v_1,\ldots, v_l, \ldots, v_m)$ with $u_1=v_l=a$, $u_j=v_1=b$ and $u_n=v_m=d_K$.  If we take the path
\[(u_{j+1},\ldots, u_n,v_{m-1},\ldots, v_l,u_2,\ldots,u_{j-1}),\]
we can extract a chordless path $p$ with starting vertex $u_{j+1}$ and ending vertex $u_{j-1}$. The path $p$ has at least three vertices because $\{u_{j+1},u_{j-1}\}$ is not in $E$. The vertex $b$ is not in $p$ because $a\neq b$, so we can add it to $p$. But then, a contradiction arises, because we obtain a chordless cycle with more than three vertices due to the fact that $b$ only forms an edge  with $u_{j-1}$ and $u_{j+1}$ among the considered vertices.
\end{proof}

\begin{lem}[Transitivity]
\label{TransitOK}
For $G=(V,E)$ a chordal graph and $K$ a maximal clique of $G$, the relation $\leqslant_K$ is transitive.
\end{lem}
\begin{proof}
For $a$, $b$ and $c$ three different vertices in $V$, suppose we have $a\leqslant_K b$ and $b\leqslant_K c$. So we have two chordless paths $(u_1,\ldots, u_j, \ldots, u_n)$ and  $(v_1,\ldots, v_l, \ldots, v_m)$ with $u_1=a$, $u_j=v_1=b$, $v_l=c$ and $u_n=v_m=d_K$. From the path $(u_1,\ldots,u_j,v_2,\ldots,v_m)$ we extract a chordless path $(u_1,\ldots, u_x,v_y,\ldots,v_m)$. If $y\in \{1,\ldots, l\}$ then $c$ is in the chordless path and we have $a\leqslant_K c$. If $y\in \{l+1,\ldots, m\}$, then from the following path:
\[(u_{j+1},\ldots, u_n,v_m,v_{m-1},\ldots, v_y,u_x,\ldots,u_{j-1}),\]
we can extract a chordless path $p$ with starting vertex $u_{j+1}$ and ending vertex $u_{j-1}$. The path $p$ has at least three vertices because $\{u_{j-1},u_{j+1}\}$ is not in $E$. If we add $b$ to $p$ we got a contradiction, because it gives a chordless cycle with more than three vertices due to the fact that $b$ only forms an edge with $u_{j-1}$ and $u_{j+1}$ among the considered vertices.
\end{proof}
 
\begin{figure}[ht]
        \centering
  \begin{tikzpicture}[scale=0.8]
\draw (0,0) node[circle,draw,fill=mydgreen,thick,inner sep=1.75pt,label={[align=center]left :{ \footnotesize \color{mydgreen} $a$, \\ \footnotesize \color{mydgreen}$u_1$, \\ \footnotesize \color{mydgreen} $ v_l$}}] (a) {};
\draw (3,0) node[circle,draw,fill=bleu,thick,inner sep=1.75pt, label=below :{\footnotesize \color{bleu}$b,u_j,v_1$}] (b) {};
\draw (6,0) node[circle,draw,fill=mydgreen,thick,inner sep=1.75pt, label={[align=center]right :{  \footnotesize \color{mydgreen}$d_K$}}] (z) {};
\draw (6,-1) node[circle,draw,fill=jaune,thick,inner sep=1.75pt, label={[align=center]right :{ \footnotesize \color{jaune} $ v_{m-1}$}}] (z2) {};
\draw (2,0) node[circle,draw,fill=bleu,thick,inner sep=1.5pt,label=above left :{\footnotesize \color{bleu}$u_{j-1}$}] (ukm1) {};
\draw (4,0) node[circle,draw,fill=bleu,thick,inner sep=1.5pt,label=above right :{\footnotesize \color{bleu}$u_{j+1}$}] (ukp1) {};
\draw[very thick, dotted, bleu] (a)--(ukm1) (ukp1)--(z);
\draw[bleu] (ukm1)--(b)--(ukp1);
\draw[very thick, dotted, jaune] (b) to[in=90,out=90] (a);
\draw[very thick, dotted, jaune] (a) to[in=-123,out=-90] (z2);
\draw[jaune] (z)--(z2);     
\draw (8,0) node[circle,draw,fill=bleu,thick,inner sep=1.75pt,label={[align=center]left :{ \footnotesize \color{bleu} $a$, \\ \footnotesize \color{bleu}$u_1$}}] (ap) {};
\draw (11,0) node[circle,draw,fill=mydgreen,thick,inner sep=1.75pt, label=below :{\footnotesize \color{mydgreen}$b,u_j,v_1$}] (bp) {};
\draw (14,1) node[circle,draw,fill=jaune,thick,inner sep=1.75pt, label={[align=center]right :{  \footnotesize \color{jaune}$v_{m-1}$}}] (z2p) {};
\draw (14,0) node[circle,draw,fill=mydgreen,thick,inner sep=1.75pt, label={[align=center]right :{ \footnotesize \color{mydgreen} $ d_K$}}] (zp) {};
\draw (10,0) node[circle,draw,fill=bleu,thick,inner sep=1.5pt,label=below left :{\footnotesize \color{bleu}$u_{j-1}$}] (ukm1p) {};
\draw (12,0) node[circle,draw,fill=bleu,thick,inner sep=1.5pt,label=below right :{\footnotesize \color{bleu}$u_{j+1}$}] (ukp1p) {};
\draw[very thick, dotted, bleu] (ap)--(ukm1p) (ukp1p)--(zp);
\node (vyp) at (12.5,1.5) {};
\draw[very thick, dotted, jaune,shorten >=-4pt] (bp) to[in=180,out=90] (vyp);
\draw[very thick, dotted, jaune,shorten <=-4pt] (vyp) to[in=123,out=0] (z2p);
\draw (8.5,0) node[circle,draw,fill=bleu,thick,inner sep=1.5pt,label=above :{\footnotesize \color{bleu}$u_x$}] (uxp) {};
\draw (11.5,1.175) node[circle,draw,fill=jaune,thick,inner sep=1.5pt,label=above left :{\footnotesize \color{jaune}$v_y$}] (cp) {};
\draw (11.15,0.707) node[circle,draw,fill=jaune,thick,inner sep=1.5pt,label=left :{\footnotesize \color{jaune}$c$}] (truecp) {};
\draw[bleu] (ukm1p)--(bp)--(ukp1p);
\draw (uxp)--(cp); 
\draw[jaune] (zp)--(z2p);
\end{tikzpicture} 
        \caption{Illustrations of the proofs of antisymmetry and transitivity for the relation $\leqslant_K$.}\label{fig:antitran}
\end{figure}
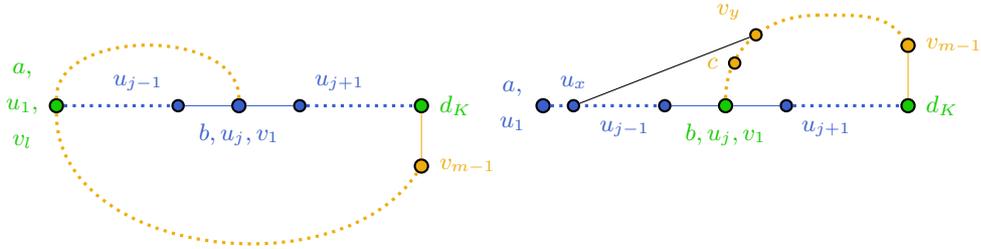

\section{Example}

\label{app:example}

Looking at Algorithms~\ref{Algo1} and~\ref{Algo2}, it seems possible to merge them to save computation time. But the situation is not that simple, because in Theorem~\ref{thm_agr}, the assumption that $a$ is $K$-blocking cannot be removed. To illustrate the mechanism of the two algorithms, we give an example. Figure~\ref{fig:explbad} shows a chordal graph $G$, its clique-separator graph $\GGG$ and a weight function. The goal is to compute $\opt(G,d_{K_1})$.

 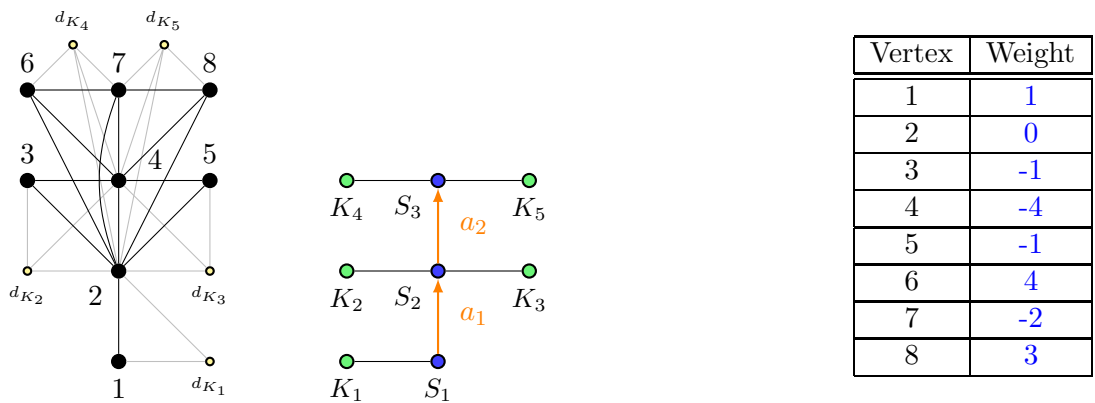
\begin{figure}[ht]
        \centering
         \begin{minipage}[c]{.74\linewidth}
  \begin{tikzpicture}[scale=0.6] 
  \tikzstyle{vertex}=[circle,draw,fill=black,thick,inner sep=1.75pt]
     \tikzstyle{vertex2}=[circle,draw,fill=yellow!50,thick,inner sep=1pt]
 \draw (-5,0) node[vertex, label=below :{$1$}]  (r1) {};
 \draw (-5,2) node[vertex, label=below left :{$2$}]  (a1) {};
 \draw (-7,4) node[vertex, label=above  :{$3$}]  (b1) {};
 \draw (-5,4) node[vertex, label={[label distance=0.15cm]5:$4$}]  (c1) {};
 \draw (-3,4) node[vertex, label=above :{$5$}]  (d1) {};
  \draw (-7,6) node[vertex, label=above :{$6$}]  (6) {};
 \draw (-5,6) node[vertex, label=above :{$7$}]  (7) {};
  \draw (-3,6) node[vertex, label=above :{$8$}]  (8) {};
 \draw (-7,2) node[vertex2, label=below :{\tiny $d_{K_2}$}]  (dk1) {};
 \draw (-3,2) node[vertex2, label=below :{{\tiny $d_{K_3}$}}]  (dk2) {};
 \draw (-3,0) node[vertex2, label=below :{{\tiny $d_{K_1}$}}]  (dk3) {};
  \draw  (-6,7) node[vertex2, label=above :{\tiny $d_{K_4}$}]  (dk4) {};
 \draw (-4,7) node[vertex2, label=above :{{\tiny $d_{K_5}$}}]  (dk5) {};
  \draw[gray!50] (r1)--(dk3)--(a1);
 \draw[gray!50] (a1)--(dk2)--(c1) (dk2)--(d1);
 \draw[gray!50] (b1)--(dk1)--(a1) (dk1)--(c1);
  \draw[gray!50] (6)--(dk4)--(7) (a1)--(dk4)--(c1);
 \draw[gray!50] (8)--(dk5)--(7) (a1)--(dk5)--(c1);
 \draw (r1)--(a1)--(b1)--(c1)--(a1)--(d1)--(c1);
 \draw (6)--(7) (6)--(c1) (6)--(a1);
 \draw (8)--(7) (8)--(c1) (8)--(a1);
 \draw (7)--(c1);
 \draw (7) to[in=110,out=-110] (a1);
 \tikzstyle{cli}=[circle,,draw,fill=vert2!75,thick,inner sep=1.75pt]
 \tikzstyle{sep}=[circle,draw,fill=blue!75,thick,inner sep=1.75pt] 
     \draw (0,0) node[cli, label=below :{\small{$K_1$}}]  (K1) {}; 
     \draw (2,0) node[sep, label=below  :{\small{$S_1$}}]  (S1) {}; 
     \draw (2,2) node[sep, label=below left:{\small{$S_2$}}]  (S2) {}; 
     \draw (2,4) node[sep, label=below left:{\small{$S_3$}}]  (S3) {}; 
     \draw (0,2) node[cli, label=below :{\small{$K_2$}}]  (K2) {}; 
     \draw (4,2) node[cli, label=below :{\small{$K_3$}}]  (K3) {}; 
     \draw (0,4) node[cli, label=below :{\small{$K_4$}}]  (K4) {}; 
     \draw (4,4) node[cli, label=below :{\small{$K_5$}}]  (K5) {};
     \draw[thick,-latex,orange] (S1)->(S2)node[midway,label=right:{$a_1$}] {};
     \draw[thick,-latex,orange] (S2)->(S3)node[midway,label=right:{$a_2$}] {};
     \draw (K1)--(S1) (K2)--(S2)--(K3) (K4)--(S3)--(K5);
\end{tikzpicture}
 \end{minipage}
 \begin{minipage}[c]{.25\linewidth}
\begin{tabular}{|c|c|}
  \hline
   Vertex & Weight  \\
  \hline
  \hline
   $1$ & {\color{blue} 1}\\  \hline
   $2$ & {\color{blue} 0}\\  \hline
   $3$ & {\color{blue} -1}\\  \hline
   $4$ & {\color{blue} -4}\\ \hline 
   $5$ & {\color{blue} -1}\\  \hline
   $6$ & {\color{blue} 4}\\  \hline
   $7$ & {\color{blue} -2}\\  \hline
   $8$ & {\color{blue} 3}\\
     \hline
\end{tabular}
 \end{minipage}
        \caption{A chordal graph and its clique-separator graph for illustrating Algorithms~\ref{Algo1} and~\ref{Algo2}.}
        \label{fig:explbad}
\end{figure}

We look at the preprocessing phase first, and we use Algorithm~\ref{Algo2} to label $a_1$ and $a_2$. Because $G\ominus a_2\subset G\ominus a_1$, we  begin by computing $\lbl(a_2)$. In other word, we want to compute a maximum-weight $S_2$-rooted convex set of $G\ominus a_2$, illustrated in Figure~\ref{fig:exp2}. There is no $K$-blocking arc in $G\ominus a_2$ for any $K$ in $\KKK_{a_2}$, we can directly use Picard's algorithm on each clique in $\KKK_{a_2}$ after temporarily changing the weight of the vertices in $S_2$ in order to impose that $S_2$ be contained in the solution. So Picard's algorithm is used two times, with $K_4$ and $K_5$, and we keep the best solution among the outputs. The result will be a convex set of weight $1$, for instance $\{2,4,6,7,8,d_{K_5}\}$, which becomes $\lbl(a_2)$.

 \begin{figure}[ht]
        \centering
   \begin{tikzpicture}[scale=0.6] 
  \tikzstyle{vertex}=[circle,draw,fill=black,thick,inner sep=1.75pt]
     \tikzstyle{vertex2}=[circle,draw,fill=yellow!50,thick,inner sep=1pt]
 \draw (-5,2) node[vertex, label=below left :{$2$}]  (a1) {};
 \draw (-5,4) node[vertex, label={[label distance=0.15cm]5:$4$}]  (c1) {};
  \draw (-7,6) node[vertex, label=above :{$6$}]  (6) {};
 \draw (-5,6) node[vertex, label=above :{$7$}]  (7) {};
  \draw (-3,6) node[vertex, label=above :{$8$}]  (8) {};
  \draw  (-6,7) node[vertex2, label=above :{\tiny $d_{K_4}$}]  (dk4) {};
 \draw (-4,7) node[vertex2, label=above :{{\tiny $d_{K_5}$}}]  (dk5) {};
  \draw[gray!50] (6)--(dk4)--(7) (a1)--(dk4)--(c1);
 \draw[gray!50] (8)--(dk5)--(7) (a1)--(dk5)--(c1);
 \draw (c1)--(a1);
 \draw (6)--(7) (6)--(c1) (6)--(a1);
 \draw (8)--(7) (8)--(c1) (8)--(a1);
 \draw (7)--(c1);
 \draw (7) to[in=110,out=-110] (a1);
 \draw[black,fill=red,fill opacity=0.50] (-5,3) ellipse (0.33cm and 1.5cm);
  \draw (-4,3) node[] {{\color{red}$S_2$}};
   \tikzstyle{cli}=[circle,,draw,fill=vert2!75,thick,inner sep=1.75pt]
 \tikzstyle{sep}=[circle,draw,fill=blue!75,thick,inner sep=1.75pt]  
     \draw (2,4) node[sep, label=below:{\small{$S_3$}}]  (S3) {}; 
     \draw (0,4) node[cli, label=below :{\small{$K_4$}}]  (K4) {}; 
     \draw (4,4) node[cli, label=below :{\small{$K_5$}}]  (K5) {};
     \draw (K4)--(S3)--(K5);
\end{tikzpicture}
        \caption{The graph $G\ominus a_2$ and its clique-separator graph.}
        \label{fig:exp2}
\end{figure}
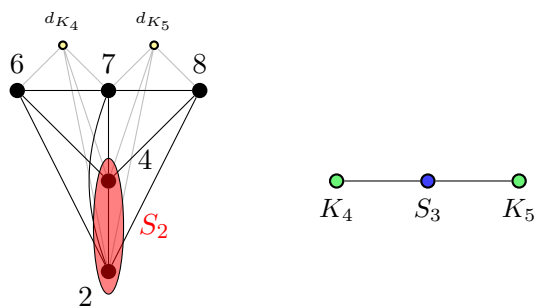

 Now we compute $\lbl(a_1)$, so we are looking for an $S_1$-rooted convex set in $G\ominus a_1$ represented in Figure~\ref{fig:exp3}.  We temporarily change the weight of the vertices in $S_1$ and we run Algorithm~\ref{Algo1} on the graph $G\ominus a_1$ four times, with the cliques $K_2$, $K_3$, $K_4$ and $K_5$. For the clique $K_4$ there is no $K_4$-blocking arc and we use Picard's algorithm.  The same process is applied with $K_5$. For the clique $K_2$, the arc $a_2$ is $K_2$-blocking but already labeled. So we identify the vertices of $(G\ominus a_2) - S_2$ to a vertex $z_{a_2}$ with weight $w(\lbl(a_2))-w(S_2)=1$. After this operation, there is no $K_2$-blocking arc and we use Picard's algorithm. Figure~\ref{fig:exp4} shows a visual representation of the transformation. The same process is applied with $K_3$.  For the labeling of $a_1$ we have used Picard's algorithm four times. A best $S_1$-rooted convex set of $G\ominus a_1$ is $\{2,6,d_{K_4}\}$ with weight $4$.

 \begin{figure}[ht]
        \centering
   \begin{tikzpicture}[scale=0.6] 
 \tikzstyle{vertex}=[circle,draw,fill=black,thick,inner sep=1.75pt]
     \tikzstyle{vertex2}=[circle,draw,fill=yellow!50,thick,inner sep=1pt]
 \draw (-5,2) node[vertex, label=below left :{$2$}]  (a1) {};
 \draw (-7,4) node[vertex, label=above  :{$3$}]  (b1) {};
 \draw (-5,4) node[vertex, label={[label distance=0.15cm]5:$4$}]  (c1) {};
 \draw (-3,4) node[vertex, label=above :{$5$}]  (d1) {};
  \draw (-7,6) node[vertex, label=above :{$6$}]  (6) {};
 \draw (-5,6) node[vertex, label=above :{$7$}]  (7) {};
  \draw (-3,6) node[vertex, label=above :{$8$}]  (8) {};
 \draw (-7,2) node[vertex2, label=below :{\tiny $d_{K_2}$}]  (dk1) {};
 \draw (-3,2) node[vertex2, label=below :{{\tiny $d_{K_3}$}}]  (dk2) {};
  \draw  (-6,7) node[vertex2, label=above :{\tiny $d_{K_4}$}]  (dk4) {};
 \draw (-4,7) node[vertex2, label=above :{{\tiny $d_{K_5}$}}]  (dk5) {};
 \draw[gray!50] (a1)--(dk2)--(c1) (dk2)--(d1);
 \draw[gray!50] (b1)--(dk1)--(a1) (dk1)--(c1);
  \draw[gray!50] (6)--(dk4)--(7) (a1)--(dk4)--(c1);
 \draw[gray!50] (8)--(dk5)--(7) (a1)--(dk5)--(c1);
 \draw  (a1)--(b1)--(c1)--(a1)--(d1)--(c1);
 \draw (6)--(7) (6)--(c1) (6)--(a1);
 \draw (8)--(7) (8)--(c1) (8)--(a1);
 \draw (7)--(c1);
 \draw (7) to[in=110,out=-110] (a1);
 \tikzstyle{cli}=[circle,,draw,fill=vert2!75,thick,inner sep=1.75pt]
 \tikzstyle{sep}=[circle,draw,fill=blue!75,thick,inner sep=1.75pt] 
     \draw (2,2) node[sep, label=below left:{\small{$S_2$}}]  (S2) {}; 
     \draw (2,4) node[sep, label=below left:{\small{$S_3$}}]  (S3) {}; 
     \draw (0,2) node[cli, label=below :{\small{$K_2$}}]  (K2) {}; 
     \draw (4,2) node[cli, label=below :{\small{$K_3$}}]  (K3) {}; 
     \draw (0,4) node[cli, label=below :{\small{$K_4$}}]  (K4) {}; 
     \draw (4,4) node[cli, label=below :{\small{$K_5$}}]  (K5) {};
     \draw[thick,-latex,orange] (S2)->(S3)node[midway,label=right:{$a_2$}] {};
     \draw  (K2)--(S2)--(K3) (K4)--(S3)--(K5);
 \draw[black,fill=red,fill opacity=0.50] (-5,2) ellipse (0.33cm and 0.33cm);
  \draw (-4,1.5) node[] {{\color{red}$S_1$}};
\end{tikzpicture}
        \caption{The graph $G\ominus a_1$ and its clique-separator graph.}
        \label{fig:exp3}
\end{figure}
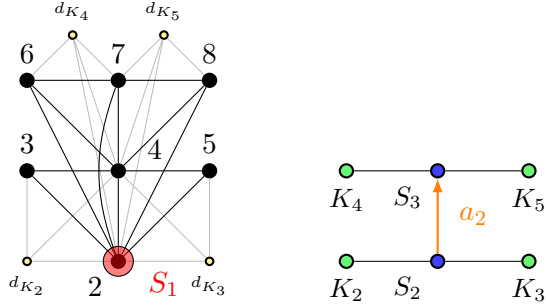

 \begin{figure}[ht]
        \centering
        \begin{minipage}[c]{.74\linewidth}
   \begin{tikzpicture}[scale=0.6] 
 \tikzstyle{vertex}=[circle,draw,fill=black,thick,inner sep=1.75pt]
     \tikzstyle{vertex2}=[circle,draw,fill=yellow!50,thick,inner sep=1pt]
 \draw (-5,2) node[vertex, label=below left :{$2$}]  (a1) {};
 \draw (-7,4) node[vertex, label=above  :{$3$}]  (b1) {};
 \draw (-5,4) node[vertex, label={[label distance=0.06cm]57:$4$}]  (c1) {};
 \draw (-3,4) node[vertex, label=above :{$5$}]  (d1) {};
  \draw (-6,6) node[vertex2, label=above :{$d_{K_6}$}]  (6) {};
  \draw (-4,6) node[vertex, label=above :{$z_{a_2}$}]  (8) {};
 \draw (-7,2) node[vertex2, label=below :{\tiny $d_{K_2}$}]  (dk1) {};
 \draw (-3,2) node[vertex2, label=below :{{\tiny $d_{K_3}$}}]  (dk2) {};
 \draw[gray!50] (a1)--(dk2)--(c1) (dk2)--(d1);
 \draw[gray!50] (b1)--(dk1)--(a1) (dk1)--(c1);
  \draw[gray!50]  (a1)--(6)--(c1) (6)--(8);
 \draw  (a1)--(b1)--(c1)--(a1)--(d1)--(c1);
 \draw   (8)--(c1) (8)--(a1);
 \tikzstyle{cli}=[circle,,draw,fill=vert2!75,thick,inner sep=1.75pt]
 \tikzstyle{sep}=[circle,draw,fill=blue!75,thick,inner sep=1.75pt] 
     \draw (2,2) node[sep, label=below :{\small{$S_2$}}]  (S2) {}; 
     \draw (2,4) node[cli, label=above:{\small{$\{d_{K_6},z_{a_2}\}\cup S_2$}}]  (S3) {}; 
     \draw (0,2) node[cli, label=below :{\small{$K_2$}}]  (K2) {}; 
     \draw (4,2) node[cli, label=below :{\small{$K_3$}}]  (K3) {}; 
     \draw  (K2)--(S2)--(K3) (S2)--(S3);
 \draw[black,fill=red,fill opacity=0.50] (-5,2) ellipse (0.33cm and 0.33cm);
  \draw (-4,1.5) node[] {{\color{red}$S_1$}};
\end{tikzpicture}
 \end{minipage}
 \begin{minipage}[c]{.25\linewidth}
\begin{tabular}{|c|c|}
  \hline
   Vertex & Weight  \\
  \hline
  \hline
   $z_{a_2}$ & {\color{blue} 1}\\ 
     \hline
\end{tabular}
 \end{minipage}
        \caption{The graph after the identification  of   vertices in order to remove $K_2$-blocking arcs.}
        \label{fig:exp4}
\end{figure}
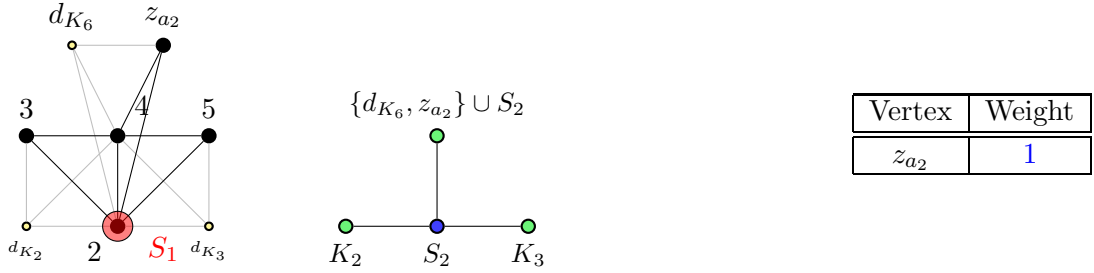

Now every arc is labeled, and we look at the computing phase. Because we want a maximum-weight $d_{K_1}$-rooted convex set, we run Algorithm~\ref{Algo1}. There is only one $K_1$-blocking arc, namely $a_1$. So we identify the vertices of $(G\ominus a_1) - S_1$ to a vertex $z_{a_1}$ with weight $w(\lbl(a_1))-w(S_1)=4$. After this operation, there is no $K_1$-blocking arc as shown in Figure~\ref{fig:exp5}, and we use Picard's algorithm to find $\{1,2,z_{a_1},d_{K_1}\}$ as a maximum-weight $K_1$-rooted convex set, which gives rise to the convex set $C^{*}=\{1,2,d_{K_1}\}\cup (\lbl(a_1) \setminus S_1)=\{1,2,6,d_{K_1},d_{K_4}\}$ of $G$, with $w(C^{*})=5$.  
 
 \begin{figure}[ht]
        \centering
         \begin{minipage}[c]{.74\linewidth}
  \begin{tikzpicture}[scale=0.6] 
  \tikzstyle{vertex}=[circle,draw,fill=black,thick,inner sep=1.75pt]
     \tikzstyle{vertex2}=[circle,draw,fill=yellow!50,thick,inner sep=1pt]
 \draw (-5,0) node[vertex, label=below :{$1$}]  (r1) {};
 \draw (-5,2) node[vertex, label=below left :{$2$}]  (a1) {};
 \draw (-5,4) node[vertex, label=above:{$z_{a_1}$}]  (c1) {};
 \draw (-3,0) node[vertex2, label=below :{{\tiny $d_{K_1}$}}]  (dk3) {};
 \draw (-7,2) node[vertex2, label=below :{{\tiny $d_{K_7}$}}]  (dk7) {};
  \draw[gray!50] (r1)--(dk3)--(a1) (c1)--(dk7)--(a1);
 \draw (r1)--(a1)--(c1);
 \tikzstyle{cli}=[circle,,draw,fill=vert2!75,thick,inner sep=1.75pt]
 \tikzstyle{sep}=[circle,draw,fill=blue!75,thick,inner sep=1.75pt] 
     \draw (0,2) node[cli, label=below :{\small{$K_1$}}]  (K1) {}; 
     \draw (2,2) node[sep, label=below  :{\small{$S_1$}}]  (S1) {}; 
     \draw (4,2) node[cli, label=below:{\small{$\{2,z_{a_1},d_{K_7}\}$}}]  (S2) {}; 
     \draw (K1)--(S1)--(S2);
\end{tikzpicture}
 \end{minipage}
 \begin{minipage}[c]{.25\linewidth}
\begin{tabular}{|c|c|}
  \hline
   Vertex & Weight  \\
  \hline
  \hline
   $1$ & {\color{blue} 1}\\  \hline
   $2$ & {\color{blue} 0}\\  \hline
   $z_{a_1}$ & {\color{blue} 4}\\  \hline
\end{tabular}
 \end{minipage}
        \caption{The state of the graph after removing the $K_1$ blocking arc.}
        \label{fig:exp5}
\end{figure}
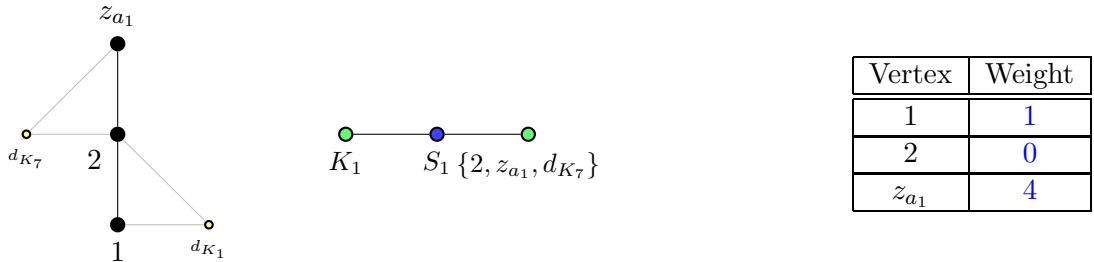

\section{Detailed time complexity}
\label{app:cmplx}

We now prove that for a chordal graph $G=(V,E)$ given by its adjacency matrix, and a weight function on the vertices, 
the running time of our algorithm is in $O(|V|^{2}|E|^{2}\log(\frac{|V|^{2}}{|E|}))$, in the worst-case. We need  to check that the information needed for the execution of Algorithms~\ref{Algo1} and~\ref{Algo2} can be computed in advance, only once for a given graph. We recall that for a graph $G$, finding a connected component of $G- X$ that contains $v$ where $X\subseteq V$ and $v\in V\setminus X$ can be done in time $O(|V|+|E|)$ as stated by Hopcroft and Tarjan~\cite{Hopcroft73}. We will use the following lemma.

\begin{lem}\label{lem:L0}
Given a (connected) chordal graph $G=(V,E)$  with clique-separator graph $\GGG$, we have $|V|-1\leqslant |E|$ and $|\mathit{Ar}_{\GGG}|\leqslant |E|$.
\end{lem}
\begin{proof}
The first inequality follows from connectedness. For the second inequality, notice each arc $(S_1,S_2) \in \mathit{Ar}_{\GGG}$ is generated by two cliques $S_1$ and $S_2$ such that $S_1\subset S_2$ and  there is no vertex $S_3$ such that $S_1 \subset S_3 \subset  S_2$. We assign to each arc $a=(S_1,S_2)$ a unique edge $\{s_1,s_2\}$ in $G$ with $s_1 \in S_1$ and $s_2\in S_2$. This implies $|\mathit{Ar}_{\GGG}|\leqslant |E|$.
\end{proof}

\begin{lem}\label{lem:L1}
Given a chordal graph $G=(V,E)$ and its clique-separator graph $\GGG$, listing $G\ominus a$ for all $a$ in $\mathit{Ar}_{\GGG}$ can be done in $O(|E|^{2})$.
\end{lem}
\begin{proof}
For each arc $a=(S_1,S_2)$, we need to find the connected component of $G - S_1$ that intersects $S_2$. Using Lemma~\ref{lem:L0}, we have $O(|E|^{2})$ as total running time.
\end{proof}

\begin{lem}\label{lem:L2}
Given a chordal graph $G=(V,E)$ with clique-separator graph $\GGG$, obtaining the list, for each  $a$ in $\mathit{Ar}_{\GGG}$, of the cliques in $\KKK_G$ for which $a$ is $K$-blocking takes  $O(|E|^{2}|V|^{2})$.
\end{lem}
\begin{proof}
By Lemma~\ref{lem:L1} we can obtain the  list of $G\ominus a$ for all $a$ in $\mathit{Ar}_{\GGG}$ in $O(|E|^{2})$. Then, for each $a=(S_1,S_2)$ in $\mathit{Ar}_{\GGG}$, and for each $K$ in $\KKK_G$, we check two conditions. First, that $d_K$ is not in $G\ominus a$ (i.e. $S_1$ is a $d_ks_2$-separator for $s_2$ in $S_2$). Second, that there is no $a'=(S_3,S_1)$ in  $\mathit{Ar}_{\GGG}$ such that $d_K$ is not in $G\ominus a'$ (i.e.  $S_1$ is minimal among the $s_2d_K$-separator for every $s_2$ in $S_2$).
\end{proof}

\begin{lem}\label{lem:L3}
Given a chordal graph $G=(V,E)$ and its clique-separator graph $\GGG$, obtaining the relations $\leqslant_K$ for all $K$ in $\KKK_G$ takes time $O(|V|^{2}|E|^{2})$.
\end{lem}
\begin{proof}
By Lemmas~\ref{lem:L1} and~\ref{lem:L2} we can obtain, for each arc $a$ in $\mathit{Ar}_{\GGG}$, the vertices in $G\ominus a$ and the cliques in $\KKK_G$ for which $a$ is $K$-blocking in $\mathit{Ar}_{\GGG}$ in $O(|E|^{2}|V|^{2})$. First, for all $K$ in  $\KKK_G$, we set $d_k\leqslant_K k$ for all $k$ in $K$. Second, for all $K$ in  $\KKK_G$, and for all $a=(S_1,S_2)$ in  $\mathit{Ar}_{\GGG}$ such that $a$ is $K$-blocking,  we set $s\leqslant_K u$ for all $s\in S_1$ and $u\in (G\ominus a) - S_1$.
\end{proof}

\begin{lem}\label{lem:L4}
Given a chordal graph $G=(V,E)$ and its clique-separator graph $\GGG$, sorting the arcs of $\mathit{Ar}_{\GGG}$  such that  $G\ominus a \subseteq G\ominus a'$ implies $a<a'$ can be done in $O(|V||E|^{2})$ time.
\end{lem}
\begin{proof}
We use Lemma~\ref{lem:L1} to obtain a list of $G\ominus a$ for all $a$ in $\mathit{Ar}_{\GGG}$. The comparison between $G\ominus a$ and $G\ominus a'$ for $a$, $a'$ in  $\mathit{Ar}_{\GGG}$ takes $O(|V|)$ time. Hence sorting takes $O(|V||E|^{2})$ time.
\end{proof}

\begin{lem}\label{lem:L5}
Given a chordal graph $G=(V,E)$ with clique-separator graph $\GGG$ and the list of $G\ominus a$ for all $a$ in $\mathit{Ar}_{\GGG}$, obtaining elements in $\KKK_a$ for all $a$ in $\mathit{Ar}_{\GGG}$ takes $O(|E||V|^{2})$ time.
\end{lem}

\begin{proof}
For each arc $a=(S_1,S_2)$ in $\mathit{Ar}_{\GGG}$, we look at each $K\in \KKK_G$ such that $S_1\subset K$, and we check if the clique is in  $G\ominus a$.
\end{proof}

\end{spacing}

\end{document}